\documentclass[journal]{IEEEtran}
\usepackage{amssymb}
\usepackage{amsthm}
\usepackage{graphicx}
\usepackage{tikz}
\usetikzlibrary[patterns]
\usepackage{pgfplots}
\pgfplotsset{compat=1.3}
\usepackage{color}
\usepackage{mathtools}
\usepackage{empheq}
\usepackage{float}
\ifCLASSINFOpdf
\else
\fi
\usepackage{amsmath}
\usepackage{algorithmic}

\usepackage{stfloats}
\usepackage{url}
\newcommand*{\pd}[3][]{\frac{\partial^{#1} #2}{\partial #3}}
\newcommand\tuu{0.2033}
\newcommand\tud{0.0328}

\newcommand\vod{1.5415}
\newcommand\vot{1.7454}
\newcommand\voq{1.8024}

\newcommand\vuu{0.578}
\newcommand\vud{0.6155}
\newcommand\vut{0.6358}
\newcommand\vuq{0.6432}

\newcommand\vmuu{1.3699}
\newcommand\vmud{1.5743}

\newcommand{\simDNormale}{\stackrel{\mathclap{d}}{\sim}\mathcal N}
\newcommand{\Sig}{\sigma^2}
\newcommand*{\Vt}[1]{\Gamma^{t_{#1}}_{t_{#1}}}

\newcommand{\R}{\mathbb R}

\DeclareMathOperator*{\argmin}{arg\,min}
\newcommand{\parallelsum}{\mathbin{\!/\mkern-5mu/\!}}

\newtheorem{definition}{Definition}
\newtheorem{lemma}{Lemma}
\newtheorem{theorem}{Theorem}
\newtheorem{prop}[theorem]{Proposition}
\newtheorem{corollary}{Corollary}

\begin{document}
\title{Optimal Measurement Times for a Small Number of Measures of a Brownian Motion over a Finite Period}

%\author{Michael~Shell,~\IEEEmembership{Member,~IEEE,}
%        John~Doe,~\IEEEmembership{Fellow,~OSA,}
%        and~Jane~Doe,~\IEEEmembership{Life~Fellow,~IEEE}% <-this % stops a space
%\thanks{M. Shell was with the Department
%of Electrical and Computer Engineering, Georgia Institute of Technology, Atlanta,
\author{Alexandre~Aksenov, Pierre-Olivier~Amblard, Olivier~Michel, Christian~Jutten
\thanks{This work has been partly supported by the European project ERC-2012-AdG-320684-CHESS.}}


\markboth{IEEE Transactions on Signal Processing,~Vol., No.}%
{Shell \MakeLowercase{\textit{et al.}}: Optimal measurement times in a Kalman filter}



\maketitle

\begin{abstract}
The measure timetable plays a critical role for the accuracy of the estimator. This article deals with the optimization of the schedule of measures for observing a random process in time using a Kalman filter, when the length of the process is finite and fixed, and a fixed number of measures are available. The measuring devices are allowed to differ.  The mean variance of the estimator is chosen as criterion for optimality. 
The cases of $1$ or $2$ measures are studied in detail, and analytical formulas are provided. 
\end{abstract}


% \fi
\IEEEpeerreviewmaketitle

\section{Introduction}
%\IEEEPARstart{T}{his} demo file is intended to serve as a ``starter file''
%for IEEE journal papers produced under \LaTeX\ using
%IEEEtran.cls version 1.8b and later.
%I wish you the best of success.

\IEEEPARstart{W}{hen} a latent phenomenon is observed through different acquisition methods, more information can be acquired than from a single method, but making the most of these measurements is a challenge \cite{Lahat,RoumeliotisBekey,HallLinas}. 
%RoumeliotisBekey,HallLinas,ViswanathanVarshney ->read
This is due to discrepancies in the nature of data, in particular in the sampling. The observer often cannot control the instants of measure and makes regular measures with each of the available sensors. In this case, controlling the delays between measurements with different sensors can lead to a consequent gain in the quality of the estimator \cite{Bourrier}. One may also ask: what is the optimal %(not necessarily regular) 
timetable of measurements when the devices are of different quality? This problem is explored in several recent papers. 

\subsection{Previous work}

Different models of the observed process and of the sensors as well as different optimization criteria have been explored. 

Models of the observed process of infinite duration have been consiedered \cite{IFAC,Bourrier}. In  this case, the mean covariance of the estimator over a long period of observation is minimized. In other terms, the optimization criterion only takes the steady-state performance of a periodic schedule into account. A model in contiuous time is explored in \cite{Bourrier}, while the time is discretized in the model of \cite{IFAC}. Another notable difference between the two models lies in that a measure is performed at every moment of the discrete time in the text \cite{IFAC}. 
As opposed to optimizing the steady-state performance \cite{IFAC,Bourrier},
local optimization is performed in the setting considered in \cite{Orihuela}. The resulting schedule is proved to be ultimately periodic, which is an a priori assumption in \cite{IFAC,Bourrier}.

When the process has a finite duration, the steady-state is not achieved (e.g., \cite{pinguins}). 
Optimizing the performance over a finite time interval is to be considered \cite{HerringMelsa,WarringtonDhir}. The optimal periodic schedule in a model with discrete time is sought in \cite{WarringtonDhir} with respect to the performance over a finite time interval. It is supposed that the interval is long enough with respect to the measurement period. 
 No additional assumptions regarding the number of measurements or the duration of the process (which is supposed to be finite) are made in the seminal work \cite{HerringMelsa}. A model, where sensors are active during an interval of time, is considered. The length of the interval of activation is a result of a tradeoff between the quality of estimation and the cost (per unit time) of using a measurement device. The optimal solution is  given in the form of an optiization problem in \cite{HerringMelsa}.

\subsection{Contributions of the paper.}

A model of observation of  a \textit{scalar} continuous latent variable on a \textit{finite} interval of time with noisy sensors is considered. Each sensor has an access to \textit{only one} measurement at one time instant. The process evolves in \textit{continuous} time in the considered model. Measurement noises of all sensors are independent random variables. 
The quality of estimation is evaluated according to the mean variance of the estimator over time.  The model studied here is simpler than that of \cite{HerringMelsa} (because the measures are instantaneous), which allows to study its properties in bigger detail. 
A qualitative study of the optimal instants of measure reveals different behaviors (``regimes") depending on the parameters. Analytic formulas for different regimes are given in the present paper and proved in the Technical Report \cite{TechnicalReport}. The optimal instant of measure is given by an analytic formula in case of one measure. In the case of two measures, an iterative algorithm and a formula in the form of a solution of a system of two equations are given.

The main theoretical results of this paper are the optimal instants of measures in the case of one or two measures (see Proposition \ref{propT1opt} and Theorem \ref{thT1T2optMonotoneCont}).
These results are illustrated by numerical computation of the optimal schedules when $2$ measures are available, the values of parameters being fixed or random.

The paper is organized as follows.  The general (multimodal, irregularly scheduled) Kalman estimation model and the cost function are defined in Section \ref{secModel}. The particular case, where the instant of only one measure is variable, has been studied in the authors' previous work \cite{art-LVAICA}.  The results of \cite{art-LVAICA} are recalled and completed in Section \ref{sec1meas}. 
The particular case, where the instants of two measures are variable, is studied in Section \ref{sec2meas}.

\section{Model Description and Optimization Objective.}
\label{secModel}
\subsection{The model of scalar Brownian Motion.}
We assume that the estimation of the system state is done by computing the time evolution of a parameter, and that the variance of the estimation grows linearly between measurements. 
This simple assumption models the fact that decreasing the measure frequency decreases the accuracy on the system state estimation.  In this purpose, we consider a real Brownian motion \(\theta(t)\) (\(t\in[0,T]\)), satisfying for $t{>}s$, \(\theta(t)-\theta(s)\simDNormale(0,\sigma^2(t-s))\)
i.e., the increments are Gaussian with mean $0$ and variance \(\sigma^2(t-s)\). 

Suppose \(n\) sensors can make measurements at moments \(t_1,\dots,t_n (0\leqslant t_1\leqslant \dots\leqslant t_n\leqslant T)\). It is assumed that each sensor \(k\) returns a measured value equal to \(X_k\) at time \(t_k\).  
No subsequence of the sequence $(t_1,\dots,t_n)$ is constrained to be regular in any sense.

Kalman filtering is used fr estimating the state $\theta(t)$ of the system using the results of the measures preceding $t$.
Suppose, the initial state $\theta(0)$ is a Gaussian random variable of mean \(\bar{\theta}_0\) and variance \(v_0\). 
Suppose that \(\theta(0)\), the measurement noise and the evolution of the Brownian motion \(\theta(t)\) are independent. 
The Kalman filter framework can apply with the state and measurement equations:
\begin{align}
\theta(t_k)&=\theta(t_{k-1})+w_k, \ w_k\simDNormale(0,\sigma^2(t_k-t_{k-1}))\label{stateKalman}\\
X_k&=\theta(t_k)+n_k, \ n_k\simDNormale(0,v_k).\label{measureKalman}
\end{align}
By the theory of Kalman filtering (see \cite{Jazwinski}), the maximum likelihood estimate \(\hat{\theta}^{t_k}_{t_k}\) of \(\theta(t_k)\) and its variance \(\Gamma^{t_k}_{t_k}\) are defined by the following recursive equations:
\begin{empheq}[left=\empheqlbrace]{align}
\hat{\theta}^{t_k}_{t_k}&=\hat{\theta}^{t_{k-1}}_{t_k}+K(t_k)\left(X_k-\hat{\theta}^{t_{k-1}}_{t_{k-1}} \right) \label{KalmanFilteringEq}\\
\hat{\theta}^{t_{k-1}}_{t_k}&=\hat{\theta}^{t_{k-1}}_{t_{k-1}} \label{KalmanPropagationEq}\\
\Gamma^{t_k}_{t_k}&=\Gamma^{t_{k-1}}_{t_k}- K(t_k) \Gamma^{t_{k-1}}_{t_k} \label{KalmanVarParallel}\\
K(t_k)&=\Gamma^{t_{k-1}}_{t_k}\left(\Gamma^{t_{k-1}}_{t_k}+v_k\right)^{-1} \label{KalmanGain}\\
\Gamma^{t_{k-1}}_{t_k}&=\Gamma^{t_{k-1}}_{t_{k-1}}+\sigma^2(t_k-t_{k-1}), \label{KalmanVarSeries}
\end{empheq}
where \(\hat{\theta}^{t_l}_{t_k}\) \((l\in \{k-1,k\})\) is the maximum likelihood estimate of \(\theta(t_k)\) conditionally to the data available at time \(t_l\), and  \(\Gamma^{t_l}_{t_k}\) is the variance of the estimate \(\hat{\theta}^{t_l}_{t_k}\). \(K(t_k)\) is the Kalman gain used for the update at time \(t_k\).
In order for \eqref{KalmanVarSeries} to make sense for \(k=1\), define \(t_0=0\) and \(\Gamma^{t_0}_{t_0}=v_0\).

Remark that, by \eqref{KalmanVarParallel},\eqref{KalmanGain}, using the fact that all quantities are scalar,
\begin{equation}
\label{RecOverlineVk0}
\Gamma^{t_k}_{t_k}=
\Gamma^{t_{k-1}}_{t_k}-\frac{\left(\Gamma^{t_{k-1}}_{t_k}\right)^2}{\Gamma^{t_{k-1}}_{t_k}+v_k}
=\frac{v_k\Gamma^{t_{k-1}}_{t_k}}{v_k+\Gamma^{t_{k-1}}_{t_k}},
\end{equation}
which is equivalent (by \eqref{KalmanVarSeries}) to
\begin{multline}
\label{RecOverlineVk}
\left(\Gamma^{t_k}_{t_k}\right)^{-1}{=}v_k^{-1}{+}\left(\Gamma^{t_{k-1}}_{t_k}\right)^{-1}=\\
v_k^{-1}{+}\left(\Gamma^{t_{k-1}}_{t_{k-1}}{+}\sigma^2(t_k{-}t_{k-1})\right)^{-1}.
\end{multline}
Therefore, each \(\Gamma^{t_k}_{t_k}\) is a rational function of \(\sigma^2,t_1,\dots,t_k,v_0,\dots,v_k\).

For each \(t\in[0,T]\), denote \(v(t)\) the variance of \(\hat{\theta}(t)\), i.e. the variance when the last measurement was taken plus the uncertainty due to the time without new feedbacks. It equals:
\begin{equation}
\label{defV}
v(t)=\Gamma^{t_k}_{t_k}+\sigma^2(t-t_k)\text{ where }k=\max\{i|t_i\leqslant t\}.
\end{equation}
\(v(t)\) is a piecewise linear function composed of line intervals of slope \(\sigma^2\). Two examples of functions \(v(t)\) are shown Figure \ref{FigvtExample2}.

\begin{figure}
\begin{center}
\begin{tabular}{c@{\hspace{1cm}}c}
\begin{tikzpicture}[scale=2]
\draw[->] (-0.2,0) -- (1.200000,0);
\draw[->] (0,-0.2) -- (0,1) node[left]{\(v(t)\)};
\draw (0.000000,1.000000) -- (0.000001,1.000001);
\draw (0.000001,1.000001) -- (0.000001,0.500000);
\draw (0.000001,0.500000) -- (0.127483,0.627482);
\draw (0.127483,0.627482) -- (0.127483,0.385554);
\draw (0.127483,0.385554) -- (0.369409,0.627480);
\draw (0.369409,0.627480) -- (0.369409,0.385553);
\draw (0.369409,0.385553) -- (0.611335,0.627480);
\draw (0.611335,0.627480) -- (0.611335,0.385553);
\draw (0.611335,0.385553) -- (1.000000,0.774218);
\draw[dashed] (0.127483,0) node[below]{$t_1$} -- (0.127483,0.385554);
\draw[dashed] (0.369409,0) node[below]{$t_2$} -- (0.369409,0.385553);
\draw[dashed] (0.611335,0) node[below]{$t_3$} -- (0.611335,0.385553);
\draw[dashed] (1,0) node[below]{$T$} -- (1,0.774218);
\end{tikzpicture}&
\begin{tikzpicture}[scale=2]
\draw[->] (-0.2,0) -- (1.200000,0);
\draw[->] (0,-0.2) -- (0,1) node[left]{\(v(t)\)};
\draw (0,0.500000) -- (0.2405,0.2405+0.5);
\draw (0.2405,0.2405+0.5) -- (0.2405,0.4255);
\draw (0.2405,0.4255) -- (0.4972,0.4255+0.4972-0.2405);
\draw (0.4972,0.4255+0.4972-0.2405) -- (0.4972,0.507);
\draw (0.4972,0.507) -- (0.641,0.507+0.641-0.4972);
\draw (0.641,0.507+0.641-0.4972) -- (0.641,0.5368);
\draw (0.641,0.5368) -- (1,0.5368+1-0.641);
\draw[dashed] (0.2405,0) node[below]{$t_1$} -- (0.2405,0.4255);
\draw[dashed] (0.4972,0) node[below]{$t_2$} -- (0.4972,0.507);
\draw[dashed] (0.641,0) node[below]{$t_3$} -- (0.641,0.5368);
\draw[dashed] (1,0) node[below]{$T$} -- (1,0.5368+1-0.641);
\end{tikzpicture}\\
\textbf{(a)}&\textbf{(b)}
\end{tabular}
\end{center}
\caption{The function \(v(t)\) in particular cases. In \textbf{(a)}, \(v_0{=}\frac12,v_1{=}v_2{=}v_3{=}1, T=1, \sigma^2=1\)  and \(t_1{=}0.128,t_2{=}0.369,t_3{=}0.611\). In \textbf{(b)}, \(v_0{=}\frac12,v_1{=}1, v_2{=}2, v_3{=}3, T=1, \sigma^2=1\)  and \(t_1{=}0.241,t_2{=}0.494,t_3{=}0.641\). The values of $v_1,v_2,v_3$ control the differences of the variance before and after the measurement. In the first example, \(v_1,v_2,v_3\) are equal, in the second example they are different.}
\label{FigvtExample2}
\end{figure}
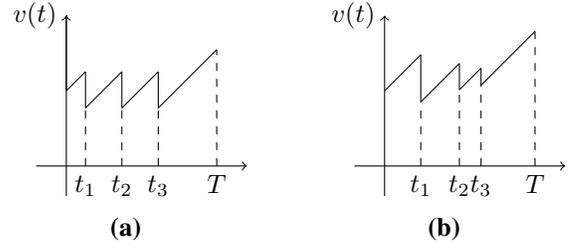

\subsection{Notation.}
Throughout this paper, the notation $(a\parallelsum b)$ will stand for $\frac{ab}{a+b}$. Note that this notation allows to rewrite \eqref{RecOverlineVk0} in a more compact way:
\begin{equation}
\label{RecOverlineVk0Compact}
\Gamma^{t_k}_{t_k}=v_k\parallelsum \Gamma^{t_{k-1}}_{t_k}.
\end{equation}

The notation $v_{k,\dots,l}$ (where $0\leqslant k\leqslant l\leqslant n$)
will stand for $(v_k\parallelsum v_{k+1}\parallelsum\dots\parallelsum v_l)$. If $k=0$, $v_{0,1,\dots,l}$ is the variance of the Kalman estimator of $\theta(0)$, which uses the information of sensors $1,\dots,l$ supposing that these sensors are activated at the instant $0$. $v_{0,1,\dots,l}$ is the smallest possible value of $\Vt{l}$. If $k>0$, $v_{k,\dots,l}$ is the error variance of the equivalent device obtained by activating the devices number $k,k+1,\dots,l$ simultaneously.

\subsection{The Optimization Criterion, General Results and Notations.}
\label{subsecCostMean}

The following optimization criterion is chosen in this article: the mean of the variance \(v(t)\) of the maximum likelihood estimator of \(\theta(t)\) is minimized by choosing the measurement instants \(t_1,\dots,t_n\). 
This implies that the following cost function is to be minimized under the constraint \(0{\leqslant} t_1{\leqslant} t_2{\leqslant}\dots{\leqslant} t_n{\leqslant} T \):
\begin{multline}
\label{defCost}
J_{\Sig,T,v_0,v_1,\dots,v_n}(t_1,\dots,t_n)= \int_0^T v(t)dt=\\
\frac{\Sig t_1^2}2+v_0t_1+
\frac{\Sig(t_2-t_1)^2}2+\Vt{1}(t_2-t_1)+\dots+\\
\frac{\Sig(T-t_n)^2}2+\Vt{n}(T-t_n).
\end{multline}
One can remark that the cost function \eqref{defCost} is rational in its \(2n{+}3\) parameters \(\Sig,T,v_0,\dots,v_n,t_1,\dots,t_n\).

If this function is minimized in a unique point 
\begin{equation}
\label{DefTkopt}
\left(t^{(n)}_{1,\text{opt}}(\Sig,T,v_0,\dots,v_n),\dots,t^{(n)}_{n,\text{opt}}(\Sig,T,v_0,\dots,v_n)\right),
\end{equation} 
these values are the optimal measurement instants. We can wonder where these instants are located, and especially if some of them are equal to zero. The minimizer is indeed unique in the cases \(n=1,2\), which is proved in Subsection \ref{ssecQuantitativeT1opt} and \ref{subsecTlargerT1crit} below.

We are also interested in the behavior of the optimal measurement times as functions of \(T\): monotonicity, asymptotic, etc.

\section{The optimal instant of one measure.}
\label{sec1meas}
\subsection{Overview of the Problem and Results.}
\label{subsecOptimalT1}

In this Section, the above problem is studied for the particular case where \(n=1\)  measure can be performed. All questions listed above are solved in terms of explicit formulas in Section \ref{ssecQuantitativeT1opt}.  Solving this particular case is necessary for tackling more complex problems. Multimodality is of smaller importance in this case, than in the more involved cases of \(n=2\) measures and \(n>2\) measures.

The cost function \eqref{defCost} takes the form
\begin{multline}
\label{fleCostOneMeasure}
J_{\Sig,T,v_0,v_1}(t_1)=\frac{\Sig t_1^2}{2}+v_0t_1+\frac{\Sig(T-t_1)^2}{2}+\\ \frac{(\Sig t_1+v_0)v_1(T-t_1)}{\Sig t_1+v_0+v_1}.
\end{multline}
Its behavior is shown Figure \ref{FigtCostMesUnique}, \textbf{(a)}. 
Remark that the RHS term in equation \eqref{fleCostOneMeasure} can be split into two terms:  the "rectangular term" \(\left(v_0t_1+\frac{(\Sig t_1+v_0)v_1(T-t_1)}{\Sig t_1+v_0+v_1}\right)\) and the "triangular term" \((\frac{\Sig t_1^2}{2}+\frac{\Sig(T-t_1)^2}{2})\), respectively accounting for the contributions of the rectangular and triangular shaped area in the integral of $v(t)$, and shown on Figure \ref{FigtCostMesUnique}, \textbf{(b)}. 
Minimizing the cost function \(J_{T,v_0,v_1}(t_1)\) constitutes a tradeoff between minimizing these two terms.

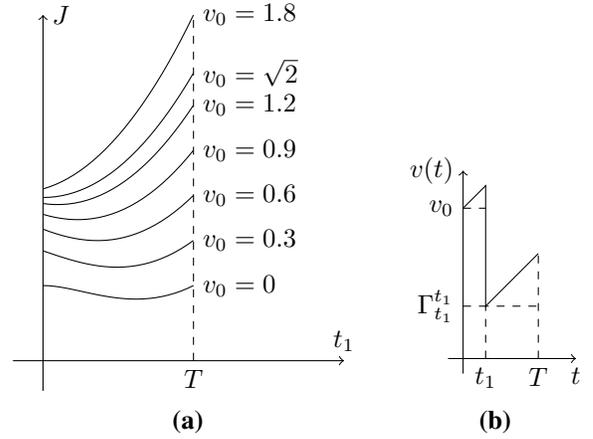
\begin{figure}
\begin{center}
\begin{tabular}{c@{\hspace{0.5cm}}c}
\begin{tikzpicture}[scale=2]
\draw[->] (-0.2,0)--(2,0) node[above]{$t_1$};
\draw[->] (0,-0.2)--(0,2.3) node[right]{$J$};
\draw[dashed] (1,0) node[below]{$T$} -- (1,2.3);

\draw[domain=0:1,variable=\x] plot ({\x},{\x^2/2 + (1-\x)^2/2 + (\x)*(1-\x)/(\x+1)});
\draw (1,0.5) node[right]{$v_0=0$};
\draw[domain=0:1,variable=\x] plot ({\x},{\x^2/2+ 0.3*\x + (1-\x)^2/2 + (\x+0.3)*(1-\x)/(\x+1.3)});
\draw (1,0.8) node[right]{$v_0=0.3$};
\draw[domain=0:1,variable=\x] plot ({\x},{\x^2/2+ 0.6*\x + (1-\x)^2/2 + (\x+0.6)*(1-\x)/(\x+1.6)});
\draw (1,1.1) node[right]{$v_0=0.6$};
\draw[domain=0:1,variable=\x] plot ({\x},{\x^2/2+ 0.9*\x + (1-\x)^2/2 + (\x+0.9)*(1-\x)/(\x+1.9)});
\draw (1,1.4) node[right]{$v_0=0.9$};
\draw[domain=0:1,variable=\x] plot ({\x},{\x^2/2+ 1.2*\x + (1-\x)^2/2 + (\x+1.2)*(1-\x)/(\x+2.2)});
\draw (1,1.7) node[right]{$v_0=1.2$};
\draw[domain=0:1,variable=\x] plot ({\x},{\x^2/2+ 1.4142*\x + (1-\x)^2/2 + (\x+1.4142)*(1-\x)/(\x+2.4142)});
\draw (1,1.9142) node[right]{$v_0=\sqrt{2}$};
\draw[domain=0:1,variable=\x] plot ({\x},{\x^2/2+ 1.8*\x + (1-\x)^2/2 + (\x+1.8)*(1-\x)/(\x+2.8)});
\draw (1,2.3) node[right]{$v_0=1.8$};
\end{tikzpicture}&
\begin{tikzpicture}
\draw[->] (-0.2,0) -- (1.5,0) node[below]{$t$};
\draw[->] (0,-0.2) -- (0,2.5) node[left]{$v(t)$};
\draw[dashed] (0,2) -- (0.3,2);
\draw[dashed] (0.3,0) node[below]{$t_1$} -- (0.3,0.697) -- (1,0.697);
\draw[dashed] (1,0) node[below]{$T$} -- (1,1.397);
\draw[dashed] (0,0.697) node[left]{$\Vt{1}$} -- (0.3,0.697);
\draw (0,2) node[left]{$v_0$} -- (0.3,2.3) -- (0.3,0.697) -- (1,1.397);
\end{tikzpicture}
\\
\textbf{(a)}&\textbf{(b)}
\end{tabular}
\end{center}
\caption{\textbf{(a)}:$J_{\Sig,T,v_0,v_1}(t_1)$ as function of $v_0$ and $t_1$. The parameters are $v_1=1, T=1, \Sig=1$. The cost function is minimized at $t_1=0$ if and only if $v_0\geqslant\sqrt{2}$. \textbf{(b)}: An example of a function \(v(t)\) showing the geometric interpretation of the rectangular and the triangular terms of the expression \eqref{fleCostOneMeasure} of the integral cost function.
}
\label{FigtCostMesUnique}
\end{figure}

Different situations are possible as it can be seen on Figure \ref{FigtCostMesUnique}, \textbf{(a)}. One can define the \textbf{regime 1} as the set of situations when \(t_1{=}0\) is the optimum. Similarly, define the \textbf{regime 2} as the set of situations where the optimal \(t_1\) is in the interior of the interval \([0,T]\). Then, the optimal \(t_1\) is the point where the derivative of the cost function \eqref{fleCostOneMeasure} vanishes. Its value is given by \eqref{t1optfonctiondeT}. Remark that in the regime 2, the optimal \(t_1\) can be larger than \(\frac{T}2\).

The optimal instant of measure is given by the following statement.
\begin{prop}
\label{propT1opt}
Let the parameters $\Sig>0, T>0, v_0\geqslant 0, v_1\geqslant 0$ be fixed. The optimal instant of measure is
\begin{multline}
\label{t1optGen}
t^{(1)}_\text{\rm opt}(\Sig,T,v_0,v_1)=\argmin_{t_1} J_{\Sig,T,v_0,v_1}(t_1)=\\\max\left(0,\frac{-3v_0{-}3v_1{+}\Sig T{+}\sqrt{(\Sig T{+}v_0{+}5v_1)^2{-}(4v_1)^2}}{4\Sig}\right).
\end{multline}
\end{prop}
Here, the general notation $t^{(1)}_\text{1,opt}$ \eqref{DefTkopt} is simplified by dropping the unnecessary index $1$.

Proposition \ref{propT1opt} is proved in Subsection \ref{ssecQuantitativeT1opt}.

\subsection{Derivation of Proposition \ref{propT1opt} and Properties of the optimal instant of measure.}
\label{ssecQuantitativeT1opt}

The behavior of the cost function can be studied using its partial derivative:
\begin{multline}
\label{DCostDt1UneMesuretprime}
\pd{J_{\Sig,T,v_0,v_1}(t_1)}{t_1}=\frac{v_0+\Sig t_1}{v_0+v_1+\Sig t_1}\times\\ \left(v_0+\Sig t_1-\Sig (T-t_1)\left(\frac{v_1}{v_0+v_1+\Sig t_1}+1\right)\right).
\end{multline}

Remark that the RHS of \eqref{DCostDt1UneMesuretprime} is a product of two increasing (with respect to \(t_1\)) factors, the first of which \(\left(\frac{v_0+\Sig t_1}{v_0+v_1+\Sig t_1}\right)\) is nonnegative (this factor vanishes iff $v_0$ and $t_1=0$). In addition, this derivative is positive in the point $t_1=T$ Therefore, the locus of positivity of \(\pd{J_T(t_1)}{t_1}\)  is an interval of the form \(]t^{(1)}_{\text{opt}},T]\), where \(t^{(1)}_{\text{opt}}\) may equal zero or be strictly positive.  
Consequently, two different behaviors of the cost function are possible. In the first case (regime $1$), it is increasing near $t_1=0$. %(its derivative at zero \eqref{DCostDt1UneMesure} is nonnegative). 
Then, the cost function $J_{\Sig,T,v_0,v_1}(t_1)$ is increasing and convex on the whole interval $[0,T]$, and its global minimum is $t^{(1)}_{\text{opt}}(T)=0$. According to \eqref{DCostDt1UneMesuretprime}, this corresponds to
\begin{equation}
\label{TcritUneMesure}
T\leqslant T^{(1)}_\text{crit}(\Sig,v_0,v_1)=\frac{v_0}{\Sig\left(\frac{v_1}{v_0+v_1}+1\right)}.
\end{equation}

In the second case (regime 2), the cost function is decreasing near \(t_1=0\). This is observed when \eqref{TcritUneMesure} does not hold, i.e. \(T\) is large or \(v_0\) is small. Then, the minimum of the cost function is reached at the only nonzero point \(t^{(1)}_\text{opt}\), where its derivative \eqref{DCostDt1UneMesuretprime} equals zero. By equating the derivative \eqref{DCostDt1UneMesuretprime} to zero, one gets the following expression for \(t^{(1)}_\text{opt}\)% and for \(T\):
\begin{equation}
\label{t1optfonctiondeT}
t^{(1)}_\text{opt}=\frac{-3v_0-3v_1+\Sig T+\sqrt{(\Sig T+v_0+5v_1)^2-(4v_1)^2}}{4\Sig}.\end{equation}
Remark that the duration $T$ can be expressed from $\Sig,t^{(1)}_\text{opt},v_0,v_1$ in this case as a rational function:
\begin{multline}
\label{TFonctiondet1opt}
T=2t^{(1)}_\text{opt}+\frac{v_0-v_1}\Sig+\frac{2v_1^2}{\Sig(v_0+\Sig t^{(1)}_\text{opt}+2v_1)}\\
 =T_{t_1}(\Sig,t^{(1)}_\text{opt},v_0,v_1).
\end{multline}

Using \eqref{TcritUneMesure} and \eqref{t1optfonctiondeT}, it is easy to check that
\begin{multline*}
t^{(1)}_\text{opt}(T^{(1)}_\text{crit})=\\\frac{-3v_0-3v_1+\Sig T^{(1)}_\text{crit}+\sqrt{(\Sig T^{(1)}_\text{crit}+v_0+5v_1)^2-(4v_1)^2}}{4\Sig}\\=0,
\end{multline*}
i.e., both formulas of regime $1$ and regime $2$ coincide if the values of the parameters lie on the boundary. This proves Proposition \ref{propT1opt}.

Remark that $T^{(1)}_\text{crit}$ is an increasing function of \(v_0\) and a decreasing function of \(v_1\) and of \(\Sig\). The limit cases of \eqref{TcritUneMesure} have the following intuitive interpretations. If $v_0\ll v_1$ (the observer has a precise knowledge about the state of the system at the instant $0$), $T^{(1)}_\text{crit}=0$, therefore the next measure should not be done in the same time. If $v_1\gg v_0$ (the measure is very inexact), the measure should be scheduled for a moment different from zero if $T\geqslant\frac{v_0}{2\Sig}$. On the other hand, if  $v_1\ll v_0$ (there is a possibility to gain precise knowledge about the system at an instant the observer can choose), then the measure should be done as soon as possible if $T\leqslant\frac{v_0}{\Sig}$.

Intuitively, ``regime $1$" is observed when \(T\) is small or \(v_0\) is large, which means that the prior information, that  the observer gets for free, is poor. In this case, it is penalizing not to take a measure immediately in order to get better information. More formally, the 
rectangular term has an order of magnitude \(O(T)\) when \(T\) tends to zero, while the triangular term has an order of magnitude \(O(T^2)\). Therefore, when \(T\) is small enough, choosing \(t_1{=}0\) should minimize both the rectangular term and the sum.

The following Proposition resumes some qualitative properties of the optimal instant of measure.
\begin{prop}
\label{propAsymptoticsT1opt}
The function $t^{(1)}_\text{\rm opt}$ is differentiable everywhere except at the border between \textbf{regime 1} and \textbf{regime 2}.
\(t^{(1)}_\text{\rm opt}(\Sig,T,v_0,v_1)\) is increasing as a function of $T$ (constant on the interval $T\in]0,T_\text{\rm crit}]$), decreasing as a function of $v_0$ and increasing as a function of $v_1$. On the interval  $T\in[T_\text{\rm crit},+\infty[$ it is a concave and strictly increasing function of $T$. Its asymptotic expansion is
\begin{equation}
\label{t1optfonctiondeTasymptote}
t^{(1)}_\text{\rm opt}(T)=\frac{\Sig T+v_1-v_0}{2\Sig}+o\left(\frac1{T}\right),
\end{equation}
the function being always smaller than its asymptote:
\begin{equation}
\label{t1optfonctiondeTIneqasymptote}
t^{(1)}_\text{\rm opt}(T)<\frac{\Sig T+v_1-v_0}{2\Sig}.
\end{equation}

When $v_1$ is large, one gets the limit:
\begin{equation}
\label{t1optv1infini}
\lim_{v_1\to\infty} t^{(1)}_\text{\rm opt}(\Sig,T,v_0,v_1)= \max\left(0,\frac{2\Sig T-v_0}{3\Sig}\right).
\end{equation}
\end{prop}
Proposition \ref{propAsymptoticsT1opt} is proved in Technical Report \cite{TechnicalReport}.

The following intuitive argument can be given for the order of magnitude of the optimal instant: \(t^{(1)}_\text{opt}(T)\sim \frac{T}2\) (by \eqref{t1optfonctiondeTasymptote}). When, \(T\) is large, the triangular term becomes more important than the ``rectangular term". Therefore, the minimum of the sum should be close to the value \(\frac{T}2\), which minimizes the triangular term.

Remark that the dependence of \(t_{1,\text{opt}}\) in \(\Sig\) and \(T\) is simplified by the relation
\begin{equation}
\label{EquivarSigmaCarre}  
t^{(1)}_\text{opt}(\frac\Sig\alpha,\alpha T,v_0,v_1)=\alpha t^{(1)}_\text{opt}(\Sig,T,v_0,v_1),
\end{equation}
therefore, the ratio \(t^{(1)}_\text{opt}/T\) depends only on \(\Sig T,v_0\) and \(v_1\).

\subsection{Bounds on the Cost Function.}
\label{ssecBounds}

One may ask for easy-to-compute  lower and upper bounds $\underline{J}$ and $\bar{J}$ of the cost function \(J\), which are independent of the instant of measure.  The value reached without measuring in the interval (which is equivalent to measuring at \(t_1=T\)) is a trivial upper bound: 
\begin{equation}
\label{upperboundCost}
J_{T,v_0,v_1}(t_1)= \int_0^T v(t)dt \leqslant v_0T+\Sig\frac{T^2}2=\bar{J}(T,v_0,v_1).
\end{equation}

A lower bound is suggested by the article \cite{Bourrier}. It leads to formulating the following. 
\begin{theorem}
\label{thresourceLBound}
The cumulative variance of a Kalman filter is bounded below by the quantity given by \eqref{resourseMinorationMeanVar}, which is independent of the instant of measure \(t_1\):
\begin{equation}
\label{resourseMinorationMeanVar}
%\frac{J_{T,v_0,v_1}(t_1)}{T}>\sqrt{\Sig v_{0,1}T}.%=\underline{J}(T,v_0,v_1).
J_{T,v_0,v_1}(t_1)>\sqrt{\Sig v_{0,1}T^3}=\underline{J}(T,v_0,v_1).
\end{equation}
\end{theorem}
Theorem \ref{thresourceLBound} is proved in Technical Report \cite{TechnicalReport}.

Two numerical experiments have been performed in order to compare the cost achieved by measuring at the optimal instant with the cost achieved by using an intuitive strategy, and with the lower bound $\underline{J}$.  Their results are shown Figure \ref{figComparisonRegular}.

In the first experiment (Figure \textbf{(a)}), the costs achieved by measuring at the optimal instant %(red curve) 
have been computed and plotted together with the costs achieved by the intuitive strategies of measuring at $0$ or at $\frac{T}{2}$, and with the corresponding values of the lower bound $\underline{J}$. The values \(T=1, \Sig=1, v_1=1\) and $v_0$ varying from $0$ to $2$ have been used for the parameters.

In the second experiment (Figure \textbf{(b)}), the costs $J_{\text{opt}}$ achieved by measuring at the optimal instant have been computed together with the costs $J_{\text{reg}}$ achieved by measuring at $\frac{T}{2}$. The values \(T=1, \Sig=1\) and $v_0,v_1$ varying from $0$ to $5$ have been used for the parameters. Figure \ref{figComparisonRegular}\textbf{(b)} shows a contour plot of the gain $\frac{J_\text{reg}-J_\text{opt}}{J_\text{reg}}$ as function of $v_0,v_1$.

Figure \ref{figComparisonRegular}\textbf{(a)} shows that measuring at the best instant among $0$ and $\frac{T}2$ leads to a performance close to the optimal. Finding the correct "regime`` is more important, therefore, than computing the optimal instant with high precision. The contour plot Figure \ref{figComparisonRegular}\textbf{(b)} shows that for parameters $v_0,v_1$ in the considered range, the gain can reach $81\%$.

\begin{figure*}
\begin{center}
\begin{tabular}{cc}
\begin{tikzpicture}
\begin{axis}[
 axis x line=center,
 axis y line=center,
xlabel=$v_0$,
xmin=0,
xmax=2.3,
ylabel=\begin{tabular}{c}\textcolor{green}{\(J_{T,v_0,v_1}(0)\)}\\%
\textcolor{blue}{\(J_{T,v_0,v_1}(\frac{T}2)\)}\\%
\textcolor{red}{\(J_{T,v_0,v_1}(t^{(1)}_\text{opt})\)}\\%
\(\underline{J}(v_0)\)
\end{tabular}
]
\input{CostsExT0.tex}
\input{CostsExToptimal.tex}
\input{CostsExTmoitie.tex}
\input{ResApprox.tex}
\end{axis}
\end{tikzpicture}&
\includegraphics[width=7cm]{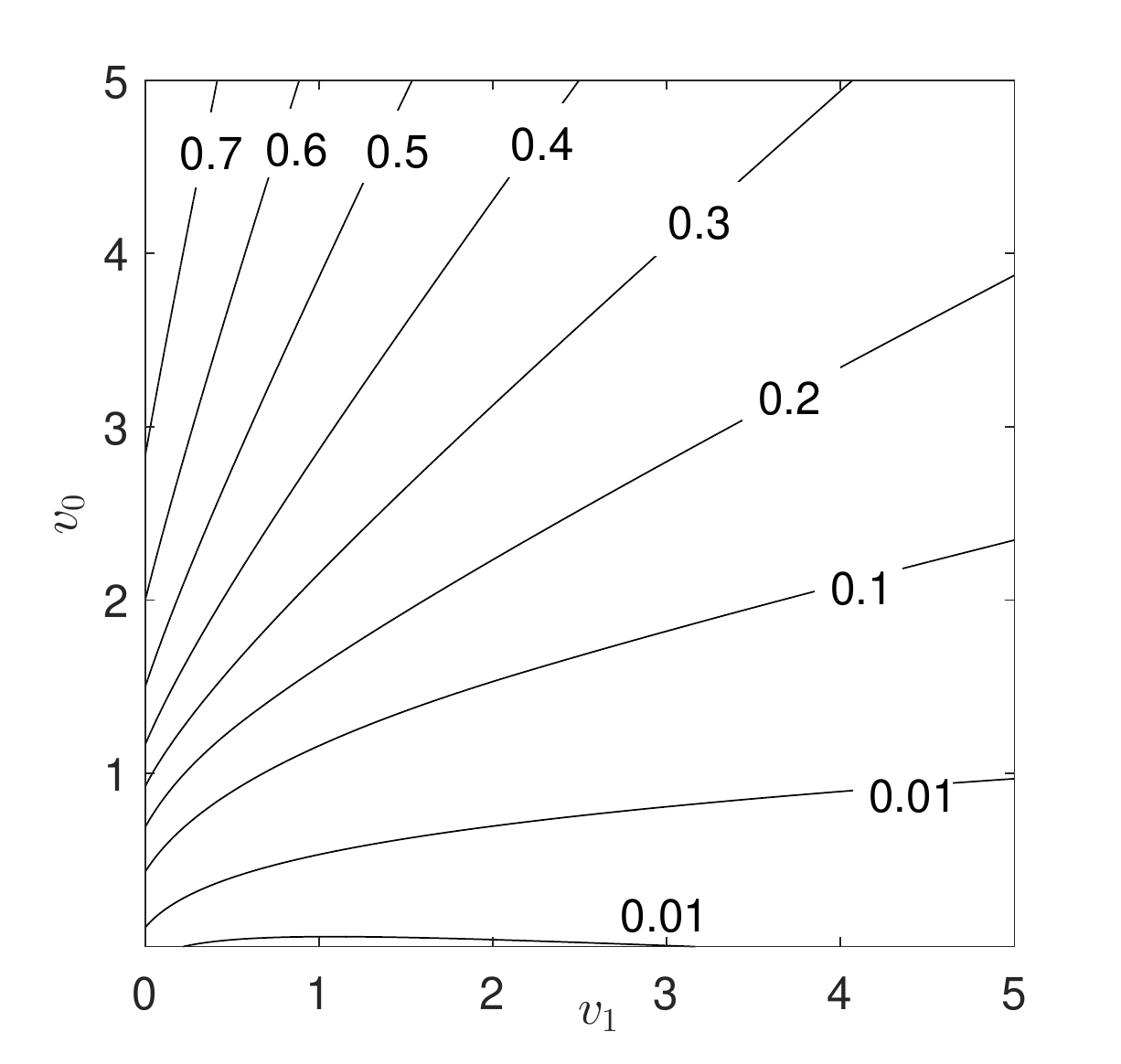}\\
\textbf{(a)}&\textbf{(b)}
\end{tabular}
\end{center}
\caption{\textbf{(a)} The costs for different choices of the instant of measurement \(t_1\) compared with the lower bound \(\underline{J}\) (see its definition \eqref{resourseMinorationMeanVar}). The parameters equal: \(T=1, \Sig=1, v_1=1\). \textbf{(b)} The contour plot of the gain of measuring at the optimal instant compared to measuring at $\frac{T}{2}$. The gain is defined as \(\frac{J_\text{reg}-J_\text{opt}}{J_\text{reg}}\). The parameters equal: \(T=1, \Sig=1, v_0,v_1\in[0,5]\).}
\label{figComparisonRegular}
\end{figure*}

\subsection{Kalman filter  with one Measure per Window, where the Windows are Periodic}
\label{subsecIterated1meas}

If only one measure is possible during a finite time interval, the optimal instant for this measure has been determined. When a Brownian motion is observed over an infinite time, the following scheduling strategy can be established: measure at moments \(t_{1,1}\in[0,T],t_{1,2}\in[T,2T],\dots,t_{1,k}\in[(k-1)T,kT],\dots\), where \(t_{1,1}\) is chosen in order to minimize the mean variance over the interval \([0,T]\), then \(t_{1,2}\) is chosen in order to minimize the mean variance over the interval \([T,2T]\) provided that the value \(v(T)\) (which depends on \(t_{1,1}\)) is used as \(v_0\) (i.e., \(T\) is the left endpoint of the interval, and \(v(T)\) is the variance of the prior information about \(\theta(T)\)), etc.
 
The parameters are: \(T,v_1\) (the error variance of every measure) and \(v_0\) (the variance of the prior information about \(\theta(0)\)). The intervals \([0,T],[T,2T],\dots\) will be called ``windows". 

The main result of this section is Theorem \ref{thPeriodicFromAperiodic}: for \(k\) big enough, \(t_{1,k}=(k-1)T\), i.e. the measures are done at the left endpoints of the corresponding ``windows".

\begin{theorem}
\label{thPeriodicFromAperiodic}
In the setting described above, the sequence of measurement instants satisfies: \(t_{1,k}=(k-1)T\) for $k$ large enough. Therefore, it is ultimately periodic. 
\end{theorem}

Theorem \ref{thPeriodicFromAperiodic} is proved in Technical Report \cite{TechnicalReport}. Figure \ref{FigvtExampleMesureUnique4iter} illustrates this setting:
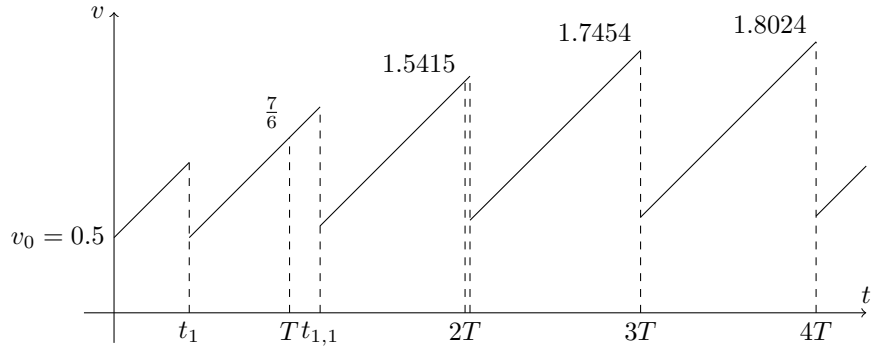
\begin{figure*}
\begin{center}
\begin{tikzpicture}[scale=2]
\draw[->] (-0.2,0) -- (5,0) node[above]{$t$};
\draw[->] (0,-0.2) -- (0,2) node[left]{$v$};
\draw (0,0.5) node[left]{$v_0=0.5$} -- (0.5,1);
\draw (0.5,0.5) -- (\vmuu,\vmuu);
\draw[dashed] (0.5,0) node[below]{$t_1$} -- (0.5,1);
\draw[dashed] (7/6,0) node[below]{$T$} -- (7/6,7/6) node[above left]{$\frac76$};
\draw[dashed] (7/6+\tuu,0) node[below]{$t_{1,1}$} -- (7/6+\tuu,\vmuu);

\draw (7/6+\tuu,\vuu) -- (7/3+\tud,\vmud);
\draw[dashed] (7/3,0) node[below]{$2T$} -- (7/3,\vod) node[above left]{$1.5415$};
\draw[dashed] (7/3+\tud,0) -- (7/3+\tud,\vmud);

\draw (7/3+\tud,\vud) -- (7/2,\vot);
\draw[dashed] (7/2,0) node[below]{$3T$} -- (7/2,\vot) node[above left]{$\vot$};

\draw (7/2,\vut) -- (14/3,\voq);
\draw[dashed] (14/3,0) node[below]{$4T$} -- (14/3,\voq) node[above left]{$\voq$};;

\draw (14/3,\vuq) -- (5,5-14/3+\vuq);

\end{tikzpicture}
\end{center}
\caption{$4$ iterations of the function \(\mathcal F_{T,v_1}\) from the initial value $v_0=\frac12$. The parameters are $v_1=1, \Sig=1, T=\frac76$. The values $v_0,v(T),v(2T),\dots$ are written on the figure.}
\label{FigvtExampleMesureUnique4iter}
\end{figure*}

One can remark that the result above resembles the results of  \cite{Orihuela}. In \cite{Orihuela}, the moments of measure are strictly periodic, while the sensor is chosen using a local optimization. On the other hand, in the present setting, the sensor cannot be chosen, while the instants of measure are chosen in periodic windows. Ultimate periodicity holds as a qualitative result in both cases.

\section{The optimal instants of two measures.}
\label{sec2meas}
\subsection{Overview of the Results.}
\label{subsec2measQualitative}

In this section, it is supposed that the observer is allowed to choose the instants $t_1$ and $t_2$ $(0\leqslant t_1\leqslant t_2\leqslant T)$ for $n{=}2$ measures with measurement noises $v_1,v_2$ respectively. Certain questions listed in Section \ref{subsecCostMean} above are answered with explicit formulas.

The cost function \eqref{defCost} can be expressed in one of the $3$ forms:
\begin{multline}
\label{fleCostTwoMeasures}
J_{\Sig,T,v_0,v_1,v_2}(t_1,t_2)=\frac{\Sig t_1^2}{2}+v_0t_1+\frac{\Sig\triangle t_2^2}{2}+\\
\frac{(\Sig t_1+v_0)v_1\triangle t_2}{\Sig t_1+v_0+v_1}+\frac{\Sig(T-t_2)^2}{2}+\Vt{2}(T-t_2)=
\\\frac{\Sig t_1^2}{2}+v_0t_1+J_{\Sig,T-t_1,\Vt{1},v_2}(t_2-t_1)=
\\
\frac{\Sig t_1^2}{2}{+}v_0t_1{+}\frac{\Sig\triangle t_2^2}{2}+\frac{(\Sig t_1{+}v_0)v_1\triangle t_2}{\Sig t_1{+}v_0{+}v_1}+\frac{\Sig(T{-}t_2)^2}{2}\\%+\frac{v_2(T-t_2)\left(\Sig(t_2-t_1)(v_1+v_0+\Sig t_1)+v_1(\Sig t_1+v_0)\right)}{\Sig t_2(v_1+v_0+\Sig t_1)-\Sig t_1(v_0+\Sig t_1-v_2)+v_0v_1+v_0v_2+v_1v_2}.
+\frac{v_2(T-t_2)\left(v_1(v_0+\Sig t_1)+\Sig\triangle t_2(v_1+v_0+\Sig t_1)\right)}{(v_1+v_2)(v_0+\Sig t_1)+v_1v_2+\Sig\triangle t_2(v_1+v_0+\Sig t_1)},
\end{multline}
where $\triangle t_2=t_2-t_1$.

It is proved (Theorems \ref{ThCritereTnuls2mes},\ref{T1sufficient},\ref{thequat1opt}) that this cost function has a unique coordinatewize local minimum which is, therefore, a global minimum. %link to {T1sufficient} 
A coordinatewize local minimum (CWLM) is defined, in an analogous way to \cite{Tseng} as follows.
\begin{definition}
\label{defCoordLocMin}
Let $f:D\subset\R^2\to\R$ be a real-valued function, and let $z=(z_1,z_2)\in D$. Then the point $z$ is called a \textit{coordinatewize local minimum} (CWLM) of $f$ if
\begin{align*}
\exists \epsilon>0&\  \forall d\in]-\epsilon,\epsilon[,\\ 
(z_1+d,z_2)\in D &\implies f(z)+(d,0)\geqslant f(z)\text{ and}\\
(z_1,z_2+d)\in D &\implies f(z)+(0,d)\geqslant f(z).
\end{align*}
\end{definition}

The argmin of $J_{\Sig,T,v_0,v_1,v_2}$ (unique) is denoted 
\begin{equation}
\label{notationT1optT2opt}
(t^{(2)}_{1,\text{opt}}(\Sig,T,v_0,v_1,v_2),t^{(2)}_{2,\text{opt}}(\Sig,T,v_0,v_1,v_2))
\end{equation}
in accordance with the general notation \eqref{DefTkopt}. 
   
One of the general remarks is that, if $t_1$ is fixed, the subproblem of determining the optimal instant $t_2$ relative to $t_1$ is reduced to determining the optimal instant of one measure (see Section \ref{sec1meas}) with the following parameters: the  length of the process is $T-t_1$, the variance of the estimate of the initial state is $\Vt{1}$, the variance of the error of the measure is $v_2$.
\begin{equation}
\label{t2optconstr}
\argmin_{t_2}J_{\Sig,T,v_0,v_1,v_2}(t_1,t_2)=t_1+t^{(1)}_{\text{opt}}(\Sig,T-t_1,\Vt{1},v_2).
\end{equation}

Finding the minimum of the cost function $J_{\Sig,T,v_0,v_1,v_2}$, studying its properties (uniqueness, position, etc) and its dependence on the parameters, such as monotonicity, continuity, is the goal of this section. An important property of the minimum is its position on the border or in the interior of the domain of definition of the function. It is sufficient to consider three qualitatively different properties of the optimal schedule (``regimes"): either $0=t_1=t_2$ (regime 1) or $0=t_1<t_2$ (regime 2) or $0<t_1\leqslant t_2$ (regime 3). Figure \ref{figExamplesJ2meas} shows examples of the cost function, which correspond to different regimes. This consideration is analogous to the one made in case of one measure.

\textbf{Regime 1} is observed when $T$ is small enough. %($T\leqslant T_{2,\text{crit}}(v_0,v_1,v_2)$). 
Then, if $t_1=0$ is fixed, the optimal instant for the second measure (determined by \eqref{t2optconstr}) is also zero. By Theorem \ref{ThCritereTnuls2mes} below, this is equivalent to saying that $(0,0)$ is the globally optimal schedule of measures.

When \textbf{regime 1} is not observed, the optimal instant for the second measure is strictly positive. One can search  the optimal schedule using the coordinate descent from $(0,0)$. The first step is finding the optimal instant of the second measure when the first measure is done at $0$ using \eqref{t2optconstr}. Call this instant $t_2^{\langle1\rangle}\in]0,T]$. On the second step, find the optimal instant of the first measure, when the second measure is done at  $t_2^{\langle1\rangle}$. Call this instant $t_{1}^{\langle1\rangle}\in[0,t_{2}^{\langle1\rangle}[$. If $t_{1}^{\langle1\rangle}=0$, the algorithm finishes and returns the schedule $(0,t_2^{\langle1\rangle})$. This situation will be called \textbf{regime 2}. By Theorem \ref{T1sufficient}, this schedule is indeed optimal.

\begin{figure*}
\begin{center}
\includegraphics[width=0.3\textwidth]{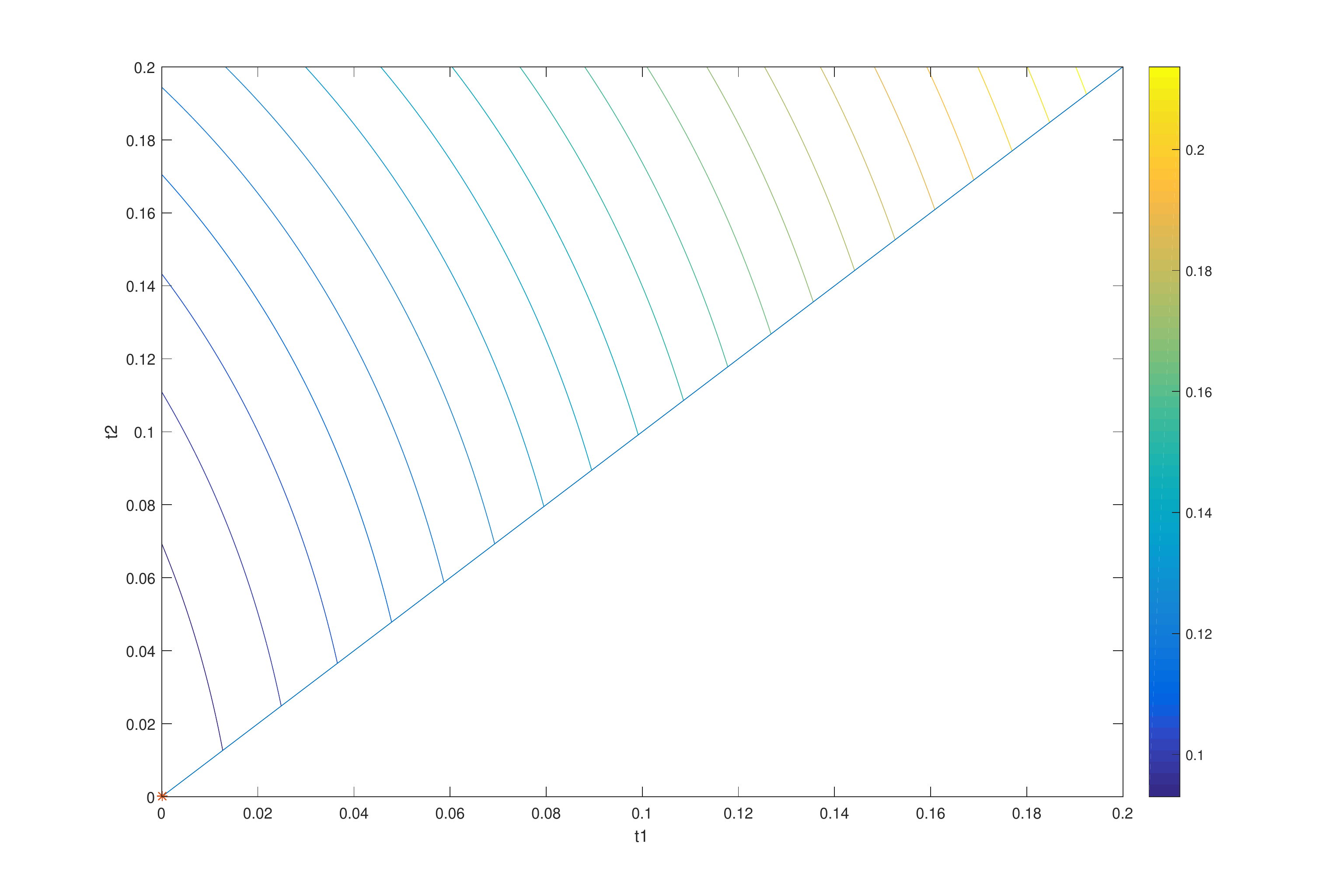}%
\includegraphics[width=0.3\textwidth]{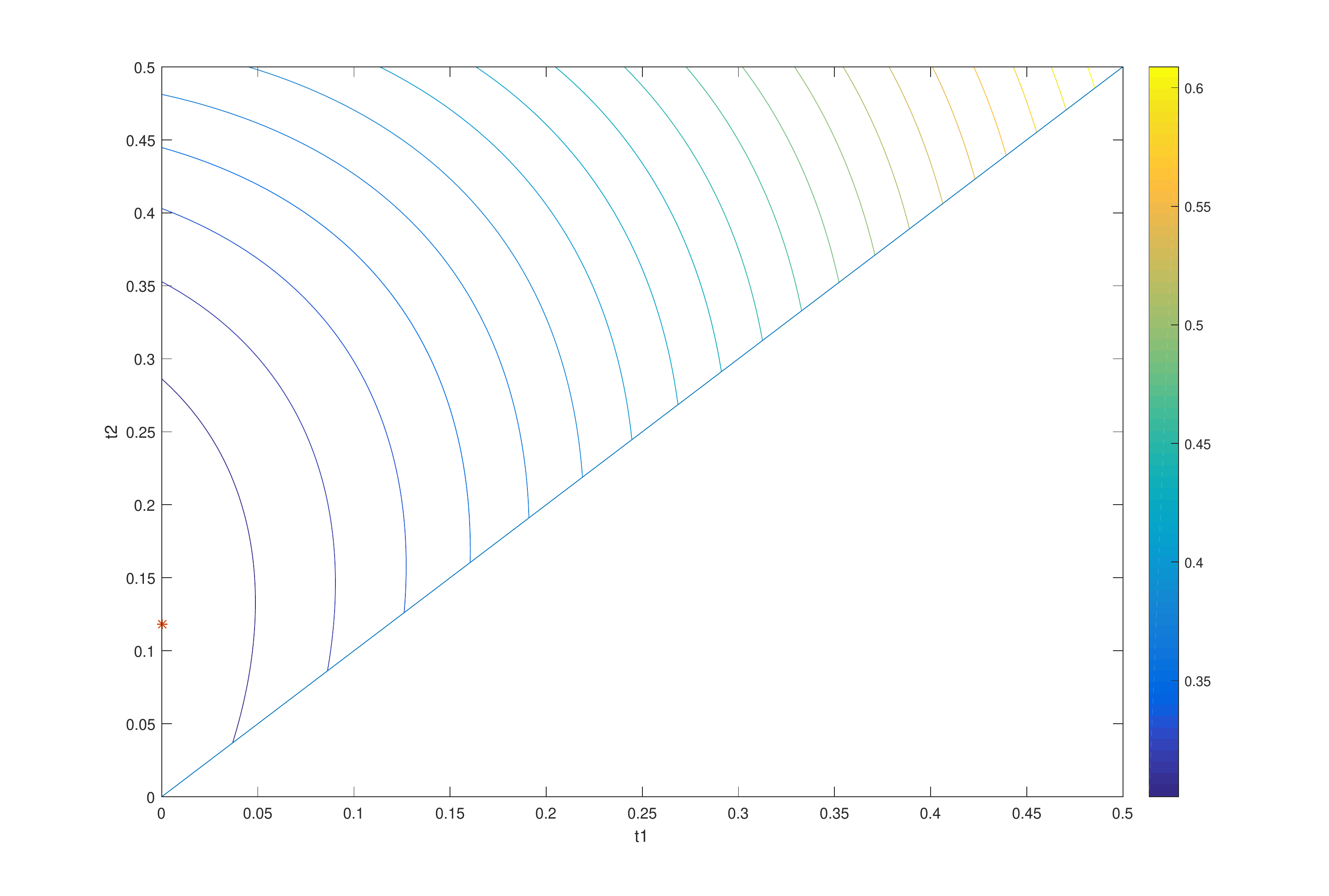}%
\includegraphics[width=0.3\textwidth]{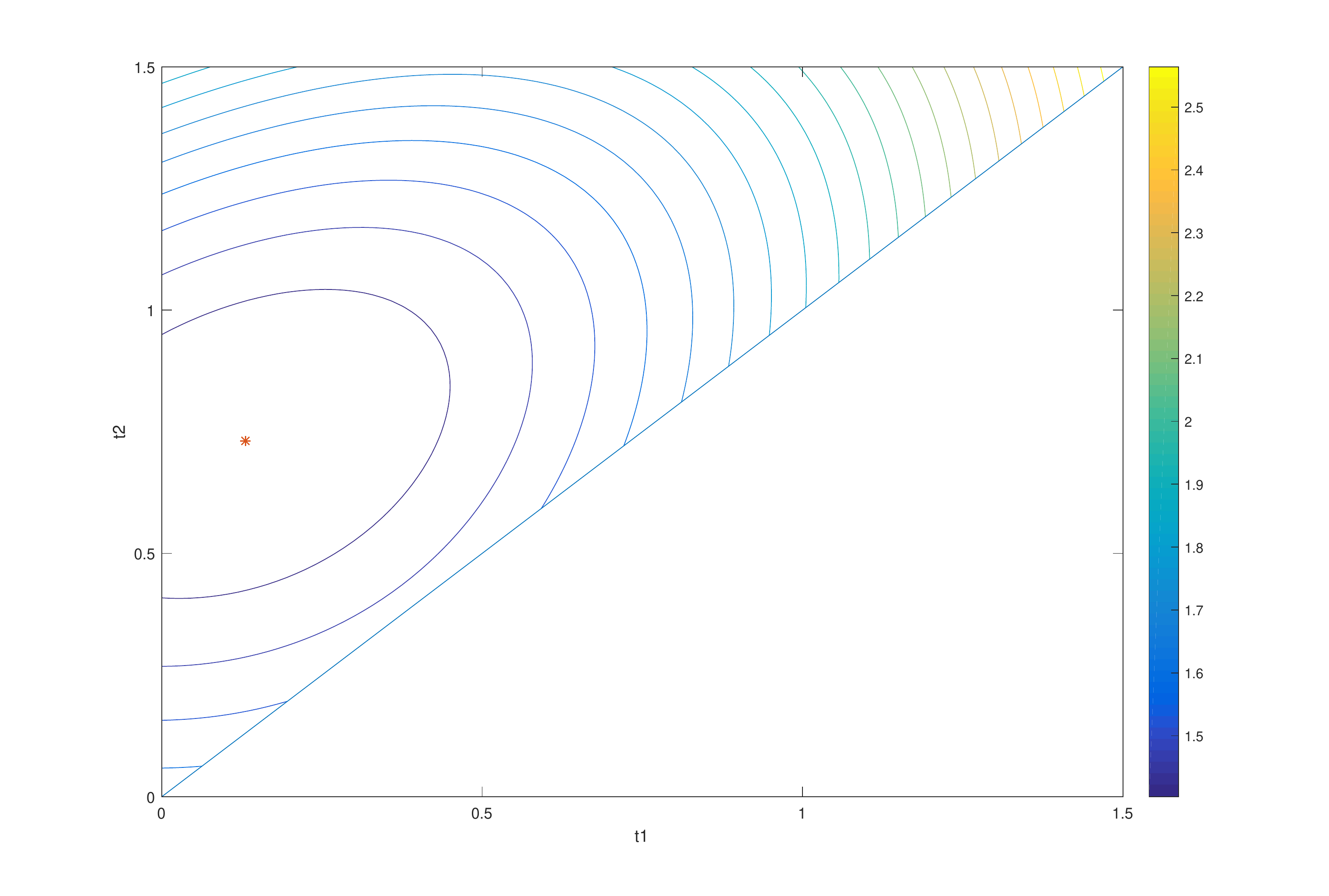}%
\end{center}
\caption{Examples of the cost function $J_{\Sig,T,v_0,v_1,v_2}(t_1,t_2)$. In all plots, $\Sig=1$ and $v_0=v_1=v_2=1$. In the first example, $T=0.2$. In the second example, $T=0.5$. In the third example, $T=1.5$.}
\label{figExamplesJ2meas}
\end{figure*}

\begin{figure*}
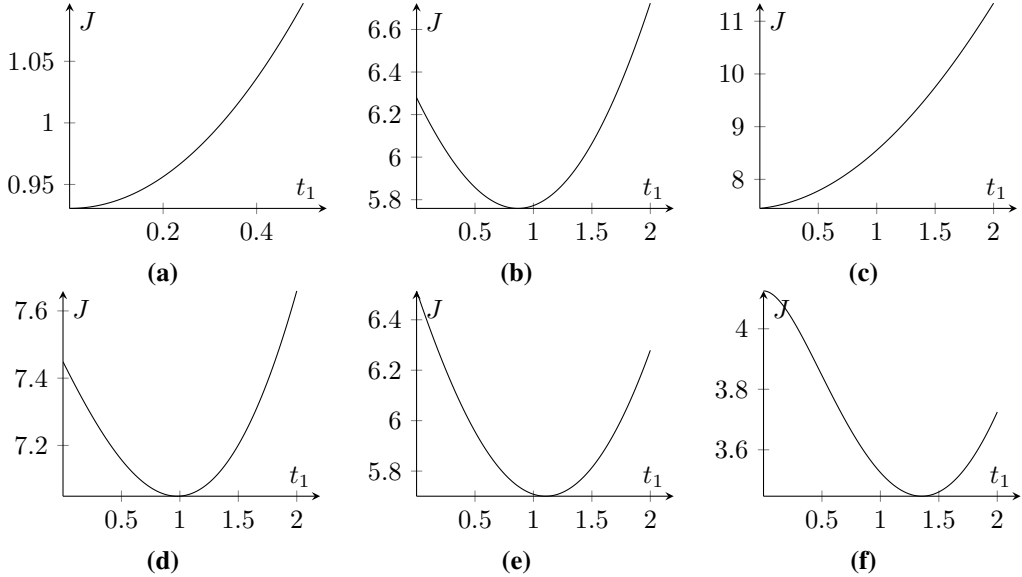

\begin{center}
\begin{tabular}{ccc}
\begin{tikzpicture}
\begin{axis}[
 axis x line=center,
 axis y line=center,
xlabel=$t_1$,
xmin=0,
xmax=0.55,
ylabel=$J$,
width=5cm]
\input{CostParams1.tex}
\end{axis}
\end{tikzpicture}&%
\begin{tikzpicture}
\begin{axis}[
 axis x line=center,
 axis y line=center,
xlabel=$t_1$,
xmin=0,
xmax=2.2,
ylabel=$J$,
width=5cm]
\input{CostParams3.tex}
\end{axis}
\end{tikzpicture}&%
\begin{tikzpicture}
\begin{axis}[
 axis x line=center,
 axis y line=center,
xlabel=$t_1$,
xmin=0,
xmax=2.2,
ylabel=$J$,
width=5cm]
\input{CostParams4.tex}
\end{axis}
\end{tikzpicture}\\
\textbf{(a)}&\textbf{(b)}&\textbf{(c)}\\
\begin{tikzpicture}
\begin{axis}[
 axis x line=center,
 axis y line=center,
xlabel=$t_1$,
xmin=0,
xmax=2.2,
ylabel=$J$,
width=5cm]
\input{CostParams5.tex}
\end{axis}
\end{tikzpicture}&%
\begin{tikzpicture}
\begin{axis}[
 axis x line=center,
 axis y line=center,
xlabel=$t_1$,
xmin=0,
xmax=2.2,
ylabel=$J$,
width=5cm]
\input{CostParams6.tex}
\end{axis}
\end{tikzpicture}&%
\begin{tikzpicture}
\begin{axis}[
 axis x line=center,
 axis y line=center,
xlabel=$t_1$,
xmin=0,
xmax=2.2,
ylabel=$J$,
width=5cm]
\input{CostParamsX.tex}
\end{axis}
\end{tikzpicture}\\
\textbf{(d)}&\textbf{(e)}&\textbf{(f)}
\end{tabular}
\end{center}
\caption{Examples of the cost function $t_1\mapsto J_{\Sig,T,v_0,v_1,v_2}(t_1,t_2)$  in the examples of Figure \ref{examples2meas}. One can observe the difference between \textbf{regime 2} (the function is increasing) and \textbf{regime 3} (the minimum is located inside the interval).}
\label{figExamplesT1toJ2meas}
\end{figure*}

In \textbf{regime 3}, the coordinate descent does not terminate after the first $2$ steps, i.e., $t_{1}^{\langle1\rangle}>0$. Then it is optimal to perform both measures in the interior of the interval $[0,T]$ (Theorem \ref{T1sufficient}). The distinction between  \textbf{regime 2} and \textbf{regime 3} can be done by computing the partial derivative with respect to $t_1$ of the cost function \eqref{fleCostTwoMeasures} at $(0,t_2^{\langle1\rangle})$ or, equivalently,  by comparing $T$ to a critical value.
%link to {T1critnecessary} 

The largest duration $T$, such that \textbf{regime 1} is observed, will be denoted $T^{(2)}_{2,\text{\rm crit}}$ (can be computed using \eqref{T2Crit}). Similarly, the largest duration $T$, such that \textbf{regime 1} is observed, will be denoted $T^{(2)}_{1,\text{\rm crit}}$ (can be computed using \eqref{eqdefT1crit}).

Figure \ref{figExamplesT1toJ2meas} shows different examples of functions $t_1\mapsto J_{\Sig,T,v_0,v_1,v_2}(t_1,t_2^{\langle1\rangle})$, which can be observed during the second step of this coordinate descent in sample situations.

\begin{figure*}
\begin{center}
\begin{tabular}{|c|c|c|c|c|c|c|c|c|c|c|}
\hline
Example&$v_0$&$v_1$&$v_2$&$t_2^{<1>}$&$T$&$t_1^{<1>}$&Regime&$T_{1,\text{crit}}$&$t^{(2)}_{1,\text{opt}}$&$t^{(2)}_{2,\text{opt}}$\\
\hline
\textbf{a)}&$1$&$1$&$1$&$0.5$&$\frac76$&$0$&Critical&$\frac76$&$0$&$0.5$\\
\hline
\textbf{b)}&$1$&$1$&$1$&$2$&$\frac{71}{18}\approx 3.94$&$0.8668$&\textbf{3}&$\frac76$&$1.0401$&$2.4092$\\
\hline
\textbf{c)}&$3$&$1$&$1$&$2$&$\frac{317}{76}\approx 4.17$&$0$&\textbf{2}&$4.6500$&$0$&$2$\\
\hline
\textbf{d)}&$1$&$3$&$1$&$2$&$\frac{317}{76}\approx 4.17$&$0.9768$&\textbf{3}&$1.1878$&$1.1211$&$2.985$\\
\hline
\textbf{e)}&$1$&$1$&$3$&$2$&$\frac{123}{34}\approx 3.62$&$1.1044$&\textbf{3}&$0.8630$&$1.1968$&$2.4269$\\
\hline
\textbf{f)}&$0$&$1$&$1$&$2$&$3.5$&$1.3538$&\textbf{3}&$0$&$1.5107$&$2.4196$\\
\hline
\end{tabular}
\end{center}
\caption{Sample examples of the problem of seeking the optimal instants of two measurements. In all examples, it is supposed that $\sigma^2=1$ and the regime $1$ is not observed. The columns $v_0,v_1,v_2,t_2^{<1>}$ are parameters, while the other columns can be computed using the formulae of the present article. Figure \ref{figExamplesT1toJ2meas} shows the functions to optimize when finding $t_1^{<1>}$ during the first step of the coordinate descent.}
\label{examples2meas}
\end{figure*}

Section \ref{sec2meas} is organized as follows.

A criterion of \textbf{regime 1}  together with a proof that the optimal schedule does not satisfy $0<t^{(2)}_\text{opt}=t^{(2)}_{2,\text{opt}}$ (Lemma \ref{lem2mesSimul}) is given in Subsection \ref{subsecSimultaneousT1T2}. The critical regime, on the border between regimes $2$ and $3$, is studied in Subsection \ref{subsecT10}. In particular, formulas in closed form are found for finding, to which regime belongs a given set of parameters $\Sig,T,v_0,v_1,v_2$.

Equations for the optimal instants in regime $3$ follow from the results of Subsection \ref{subsecT10}. These are discussed in Section \ref{subsecTlargerT1crit}. Some properties of the optimal instants are deduced from these equations.   

The coordinate descent algorithm can be used for finding the optimal measurement instants in \textbf{regime 3}. It is shown that this algorithm cannot converge to a point different from the global minimum of the cost function. This follows from the uniqueness of a %coordinatewize local minimum 
CWLM of the cost function $J(t_1,t_2)$ (Theorem \ref{thUniqueCoordMin}, Section \ref{subsecTlargerT1crit}).

\subsection{Strategy of proof.}
\label{subsecProofStrat}
Proving the uniqueness of a 
CWLM of the cost function $J(t_1,t_2)$ is done by considering first the borders of its domain of definition, then the interior. The border $t_1=t_2$ (represented by the diagonal in the plots Figure \ref{figExamplesJ2meas}) is studied in  Subsection \ref{subsecSimultaneousT1T2}. The border $t_1=0$ (represented by the left side in the plots Figure \ref{figExamplesJ2meas}) is studied in Subsection \ref{subsecT10}. The schedules on the border $t_2=T$ can be improved upon by decreasing $t_2$ according to the results relative to one measure. The interior is studied in  Subsections \ref{subsecT10} and \ref{subsecTlargerT1crit} using the previous results.
%refs to subsubsections

\subsection{Simultaneous measurements ($t_1=t_2$).}
\label{subsecSimultaneousT1T2}

Taking both measures at the same time makes them equivalent to a single measure of smaller error variance $v_{1,2}$. Therefore, the performance of such schedule is the same as one achieved by one measure. Lemma \ref{lem2mesSimul} shows that, except the case where the measures are at the instant $0$, such schedule can be improved upon by a small displacement of the instant of one measure. The rest of this subsection is devoted to studying the optimality of taking both measures at $0$ (\textbf{regime 1}).

\begin{lemma}
\label{lem2mesSimul}
Consider the cost function \eqref{fleCostTwoMeasures} defined on the triangular domain \({\mathcal T}_T=\{(t_1,t_2) \text{ s.t.} 0\leqslant t_1\leqslant t_2\leqslant T\}\). Let $0<t_1<T$. Then the point $(t_1,t_1)$ is not a coordintewize local minimum of $J_{\Sig,T,v_0,v_1,v_2}(t_1,t_2)$.
\end{lemma}
Lemma \ref{lem2mesSimul} is proved in %Appendix \ref{SsecProof2mesSimul}. 
Technical Report \cite{TechnicalReport}. It corresponds to the intuitive idea that the instants of measure have a tendency to ``repulse" each other.

The following criterion for deciding whether both optimal instants equal zero (\textbf{regime 1}) extends the criterion \eqref{TcritUneMesure} from the case of one measure to the case of two measures.

\begin{theorem}
\label{ThCritereTnuls2mes}
The global minimum of the cost function \eqref{fleCostTwoMeasures} is reached at the point $(0,0)$ if and only if
\begin{multline}
\label{T2Crit}
T\leqslant T^{(2)}_{2,\text{\rm crit}}(\Sig,v_0,v_1,v_2)=T^{(1)}_{\text{\rm crit}}(\Sig,v_{0,1},v_2)=\\\frac{v_{0,1}}{\Sig\left(\frac{v_2}{v_{0,1}+v_2}+1\right)}.
\end{multline}
Moreover, when \eqref{T2Crit} holds, the point $(0,0)$ is the unique %coordinatewize 
%local minimum 
CWLM of the function \eqref{fleCostTwoMeasures}.
\end{theorem}
%Theorem \ref{ThCritereTnuls2mes} is proved in Technical Report \cite{TechnicalReport}.
\begin{proof}
\textit{Direct part}. Suppose the minimum is at $(0,0)$. 
In particular, the function
\begin{equation}
\label{proofT2critDirect}
t_2\mapsto J_{\Sig,T,v_0,v_1,v_2}(0,t_2)=J_{\Sig,T,v_{0,1},v_2}(t_2)
\end{equation}
has its minimum  at $t_2=0$, therefore the regime 1 in the sense of a single measure is observed (cf Subsection \ref{ssecQuantitativeT1opt}). Therefore, the criterion \eqref{TcritUneMesure} applied to the parameters $\Sig,T,v_{0,1},v_2$ is valid. This is \eqref{T2Crit}.

\textit{Inverse part}. Suppose $T\leqslant T^{(2)}_{2,\text{crit}}$. 
For any $t_1\in[0,T[$ the minimum of the function
\begin{multline}
\label{proofT2critInverse}
[t_1,T]\ni t_2\mapsto J_{\Sig,T,v_0,v_1,v_2}(t_1,t_2)=\\\frac{\Sig t_1^2}{2}+v_0t_1+J_{\Sig,T-t_1,\Vt{1},v_2}(t_2-t_1)
\end{multline} 
is the same as the minimum of
\begin{equation}
\label{proofT2critInverse2}
[t_1,T]\ni t_2\mapsto J_{\Sig,T-t_1,\Vt{1},v_2}(t_2-t_1)
\end{equation}
and can be found using the results of Subsection \ref{ssecQuantitativeT1opt}. More precisely, regime $1$ is observed. Indeed, as the function $v_0\mapsto T^{(1)}_{\text{crit}}(\Sig,v_0,v1)$ is increasing and $\Vt{1}\geqslant v_{0,1}$, one has
\begin{equation}
\label{proofT2critInverse3}
T-t_1\leqslant T\leqslant T^{(1)}_{\text{crit}}(\Sig,v_{0,1},v_2)\leqslant T^{(1)}_{\text{crit}}(\Sig,\Vt{1},v_2).
\end{equation}  
This implies that, under the hypothesis $T\leqslant T^{(2)}_{2,\text{crit}}$, all %coordinatewize local minima 
CWLM's of the cost function $J_{\Sig,T,v_0,v_1,v_2}(t_1,t_2)$ are points of the type $(t_1,t_1)$, that is, on the diagonal. By Lemma \ref{lem2mesSimul}, the only candidate for being a %coordinatewize local minimum 
CWLM of $J_{\Sig,T,v_0,v_1,v_2}(t_1,t_2)$ is the point $(0,0)$, which is, therefore, its global minimum.

\end{proof}

Remark that the critical duration $T^{(2)}_{2,\text{crit}}$ is an increasing function of $v_0$ and of $v_1$, a decreasing function of $\Sig$ and of $v_2$.

It is proved in this subsection that a CWLM on the diagonal can only be achieved at $(0,0)$. Moreover, \eqref{T2Crit} provides a necessary and sufficient condition (depending on the parameters), which allows one to check whether $(0,0)$ is indeed a CWLM. This is equivalent to \textbf{regime 1}.

\subsection{The boundary $t_1=0$.}
\label{subsecT10}

If \eqref{T2Crit} does not hold, consider the boundary $t_1{=}0$. Taking the first measure at zero leads to the same performance as the setting with one measure of error variance $v_2$ and initial information of smaller error variance $v_{0,1}$. Theorem \ref{T1sufficient} shows that the optimal schedule is of this type for some values of parameters. This subsection is devoted to studying when this is satisfied.

The following result answers the question, whether the minimum is located on the boundary.
\begin{theorem}
\label{T1sufficient}
The global minimum of the cost function $J_{\Sig,T,v_0,v_1,v_2}$ is located on the line $(0,\cdot)$ iff $T\leqslant T^{(2)}_{1,\text{\rm crit}}(\Sig,v_0,v_1,v_2)$, where
\begin{equation}
\label{eqdefT1crit}
T^{(2)}_{1,\text{\rm crit}}(\Sig,v_0,v_1,v_2)=T_{t_1}(\Sig,t_{2,1,\text{\rm crit}}(\Sig,v_0,v_1,v_2),v_{0,1},v_2).%\\>T^{(2)}_{2,\text{\rm crit}}(\Sig,v_0,v_1,v_2).
\end{equation}
Here, the function $T_{t_1}$ is defined by 
\begin{equation}
\label{TFonctiondet1optRappel}
T_{t_1}(\Sig,t_1,v_0,v_1)=2t_1+\frac{v_0-v_1}\Sig+\frac{2v_1^2}{\Sig(v_0+\Sig t_1+2v_1)}
\end{equation}
and $\Sig t_{2,1,\text{\rm crit}}(\Sig,v_0,v_1,v_2)$ is the largest root of the equation
\begin{equation}
\label{equaT21crit}
Ax^3+Bx^2+Cx+D=0
\end{equation} 
with coefficients
\begin{align}
A(v_0,v_1,v_2)&=-(v_0+v_1)^2(v_0+2v_1),\label{defAt21crit}\\
B(v_0,v_1,v_2)&=(v_0+v_1)\times\notag\\
&(v_0^3-3((v_0+v_1)(v_0+2v_1)v_2+v_0v_1^2)),\label{defBt21crit}\\
C(v_0,v_1,v_2)&=v_2B(v_0,v_1,v_2)+\notag\\
&v_0^2(2v_0+3v_1)(v_0v_1+v_0v_2+v_1v_2),\label{defCt21crit}\\
D(v_0,v_1,v_2)&=v_0^2(v_1+v_2)(v_0+v_1)\times\notag\\
&(v_0v_1+2v_0v_2+3v_1v_2).\label{defDt21crit}
\end{align}
If $T\leqslant T^{(2)}_{1,\text{crit}}(\Sig,v_0,v_1,v_2)$, the minimum of the cost function is located  at the point $(0,t_2^{\langle1\rangle})$, where 
\begin{equation}
\label{deft2iter1}
t_2^{\langle1\rangle}=t^{(1)}_{\text{\rm opt}}(\Sig,T,v_{0,1},v_2)
\end{equation}
according to the more general equation \eqref{t2optconstr}.
\end{theorem}
Theorem \ref{T1sufficient} is proved in the Technical Report \cite{TechnicalReport}.

According to Theorems \ref{T1sufficient} and \ref{ThCritereTnuls2mes}, the optimal schedule is of the form $(0,\cdot)$, but not $(0,0)$ if and only if
\begin{equation}
T^{(2)}_{2,\text{\rm crit}}(\Sig,v_0,v_1,v_2)<T\leqslant T^{(2)}_{1,\text{\rm crit}}(\Sig,v_0,v_1,v_2),
\end{equation}
where $T^{(2)}_{1,\text{\rm crit}}$ is defined by \eqref{eqdefT1crit}-\eqref{defDt21crit} and $T^{(2)}_{2,\text{\rm crit}}$ is defined by \eqref{T2Crit}.  This case will be called \textbf{regime 2}.

The proof of Theorem \ref{T1sufficient} immediately leads to the following corollaries.
\begin{corollary}
\label{UniqueCWLMregime2}
If $T\leqslant T^{(2)}_{1,\text{\rm crit}}(\Sig,v_0,v_1,v_2)$, the point $(0,t_2^{\langle1\rangle})$ is the only CWLM of the cost function.
\end{corollary}
\begin{corollary}
\label{CriticalDurationsIncreasing}
Let $\Sig,v_1,v_2\in\R_+^*$. Then, the critical durations $T^{(2)}_{1,\text{\rm crit}},T^{(2)}_{2,\text{\rm crit}}$ as well as the duration $t_{2,1,\text{\rm crit}}$, appearing in the formulation of Theorem \ref{T1sufficient}, are strictly increasing functions of $v_0$.
\end{corollary}

The quantity $t_{2,1,\text{\rm crit}}$, appearing in the formulation of Theorem \ref{T1sufficient} has the following interpretation: the optimal schedule in the case of the duration $T=T^{(2)}_{1,\text{crit}}(\Sig,v_0,v_1,v_2)$ is $(0,t_{2,1,\text{\rm crit}})$.

\subsection{The general case and its Properties.}
\label{subsecTlargerT1crit}

Suppose that the process is long enough, i.e. $T>T^{(2)}_{1,\text{crit}}(\Sig,v_0,v_1,v_2)$. Therefore, \textbf{regime 3} is observed (Theorem \ref{T1sufficient}). In this subsection, equations for determining the optimal instants of measure $t^{(2)}_{1,\text{opt}}$ and $t^{(2)}_{2,\text{opt}}$ will be derived. 
Nontrivial optimal instants, which cannot be computed using formulae for $1$ measure, are defined in this section, and some properties of these optimal instants are proved. 

\begin{theorem}
\label{thequat1opt}
Suppose that the length of the process is larger than the critical durations: $T>T^{(2)}_{1,\text{\rm crit}}(\Sig,v_0,v_1,v_2)$. Then, the cost function $J_{\Sig,T,v_0,v_1,v_2}$ has a unique %coordinatewize local minimum 
CWLM $(t^{(2)}_{1,\text{\rm opt}},t^{(2)}_{2,\text{\rm opt}})$, and it satisfies
\begin{align}
t^{(2)}_{2,\text{\rm opt}}-t^{(2)}_{1,\text{\rm opt}}&=t_{\text{\rm opt}}^{(1)}(\Sig,T-t_{1,\text{\rm opt}}^{(2)},(v_0+\Sig t_{1,\text{\rm opt}}^{(2)})\parallelsum v_1,v_2),\label{t1t2optEqSndMeas}\\
t^{(2)}_{2,\text{\rm opt}}-t^{(2)}_{1,\text{\rm opt}}&=t_{2,1,\text{\rm crit}}(\Sig,v_0+t^{(2)}_{1,\text{\rm opt}},v_1,v_2),\label{t1t2optEq2meas}
\end{align}
where the function $t_{2,1,\text{\rm crit}}$ is defined %Definition \ref{defT21crit}, 
as the largest real root of the equation \eqref{equaT21crit} with coefficients \eqref{defAt21crit}-\eqref{defDt21crit},
and the function $t_{\text{\rm opt}}^{(1)}$ is defined Equation \eqref{t1optGen}. Moreover, the system of equations \eqref{t1t2optEqSndMeas},\eqref{t1t2optEq2meas} has a unique solution with respect to the variables $t^{(2)}_{1,\text{\rm opt}},t^{(2)}_{2,\text{\rm opt}}-t^{(2)}_{1,\text{\rm opt}}\in\R_+^*$.
\end{theorem}
Theorem \ref{thequat1opt} is proved in Technical Report \cite{TechnicalReport}.

The system of equations \eqref{t1t2optEqSndMeas},\eqref{t1t2optEq2meas} is of the form
\begin{equation*}
\left\{\begin{aligned}y&=I^{(1)}_{\Sig,T,v_0,v_1,v_2}(x)\\y&=I^{(2)}_{\Sig,v_0,v_1,v_2}(x)\end{aligned}\right.
\end{equation*}
where $x=t^{(2)}_{1,\text{\rm opt}}$ and $y=t^{(2)}_{2,\text{\rm opt}}-t^{(2)}_{1,\text{\rm opt}}$. It is interesting to study the behavior of the functions $I^{(1)}$ and $I^{(2)}$, which appear in this system. Their full definitions are
\begin{equation}
\label{funcEqt2minust11meas}
I^{(1)}_{\Sig,T,v_0,v_1,v_2}(\tilde{t}_1)=
\left\{\begin{array}{l}0\text{ if $\tilde{t}_1\geqslant T$}\\t_{\text{opt}}^{(1)}(\Sig,T-\tilde{t}_{1},(v_0+\Sig \tilde{t}_{1})\parallelsum v_1,v_2)\\\text{ otherwize.}\end{array}\right.
\end{equation}
and
\begin{equation}
\label{funcEqt2minust12meas}
I^{(2)}_{\Sig,v_0,v_1,v_2}(\tilde{t}_1)=t_{2,1,\text{\rm crit}}(\Sig,v_0+\Sig \tilde{t}_1,v_1,v_2).
\end{equation}
Here, the continuation of the function $I^{(1)}(\tilde{t}_1)$ by zero for large values serves a purely technical purpose.

The function $I^{(2)}$ has the following interpretation. Given the parameters $\Sig,v_0,v_1,v_2$, 
to each $\tilde{t}_{1}\in\R_+$, a unique duration $T$ is associated such that $\tilde{t}_1$ is the optimal instant of the first measure. Then, $I^{(2)}(t_1)$ is the distance between the optimal instants for the duration $T$. By Corollary \ref{CriticalDurationsIncreasing}, it is strictly increasing.

The function $I^{(1)}$ has a simpler definition and its interpretation is: given $\Sig,T,v_0,v_1,v_2$, it associates to each $\tilde{t}_1$ (suboptimal in general) the best interval $\tilde{t}_2-\tilde{t}_1$ between the measures. The function $I^{(1)}$ ``selects"  the point associated to the given length $T$ on the graph of $I^{(2)}_{\Sig,v_0,v_1,v_2}$. It is decreasing by Proposition \ref{propAsymptoticsT1opt}. The optimal schedule corresponds to the intersection point of the graphs of these functions.

Figure \ref{figSystemEquations} shows an example of the behavior of the functions $I^{(1)}$ and $I^{(2)}$.

\begin{figure}
\begin{center}
\includegraphics[width=0.48\textwidth]{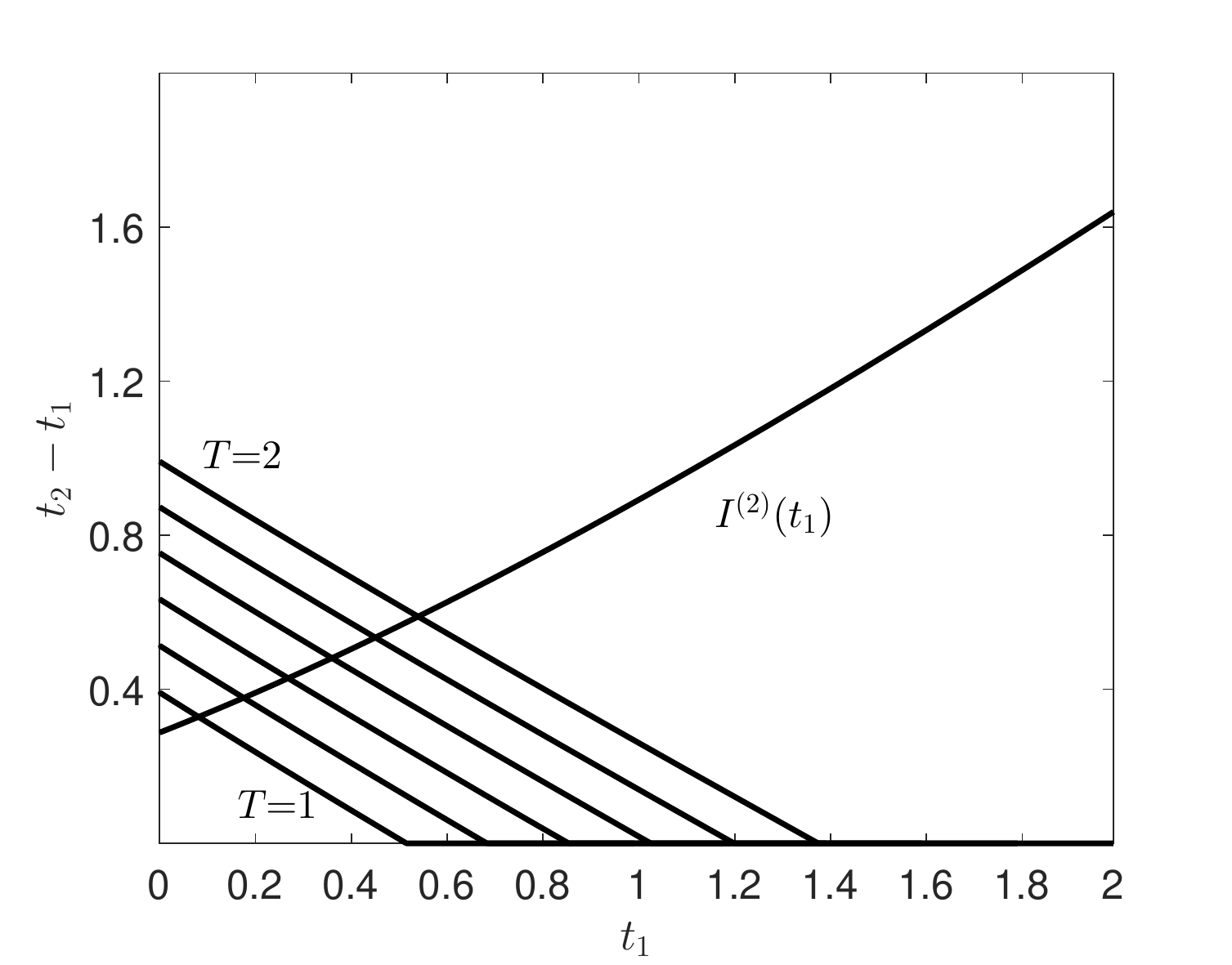}%
\end{center}
\caption{The functions $I^{(2)}(t_1)$ and $I^{(1)}_T(t_1)$. The fixed parameters are $\Sig=1,v_0=1,v_1=2,v_3=3$. $T$ takes values $1,1.2,1.4,1.6,1.8,2$. The optimal instant of first measure $t^{(2)}_{1,\text{\rm opt}}$ is the abscissa of the point of intersection. The distance between the measures in the optimal schedule is the ordinate of the point of intersection.}
\label{figSystemEquations}
\end{figure}

The next theorems assemble results for all three regimes and answer to conjectures announced in Sections \ref{subsecCostMean} and \ref{subsec2measQualitative}. 
\begin{theorem}
\label{thT1T2optMonotoneCont}
Let $\Sig,v_0,v_1,v_2\in\R_+^*$. For each $T\geqslant 0$, the minimizer of the cost function $J_{\Sig,T,v_0,v_1,v_2}$ is unique.

The functions $T\mapsto t^{(2)}_{1,\text{\rm opt}}(\Sig,T,v_0,v_1,v_2)$ and $T\mapsto t^{(2)}_{2,\text{\rm opt}}(\Sig,T,v_0,v_1,v_2)$ are continuous and monotonically increasing.
\end{theorem}

Theorem \ref{thT1T2optMonotoneCont} is proved in the Technical Report \cite{TechnicalReport}.

Moreover, the cost functions $J(t_1,t_2)$ have unique CWLM's.%coordinatewize local minima.

\begin{theorem}
\label{thUniqueCoordMin}
Let $\Sig,v_0,v_1,v_2,T\in\R_+^*$. Then, the function $J_{\Sig,T,v_0,v_1,v_2}(t_1,t_2)$ has a unique CWLM. 
\end{theorem}
\begin{proof}
The theorem follows from Theorem \ref{ThCritereTnuls2mes} (in case of \textbf{regime 1}), Theorem \ref{T1sufficient} (in case of \textbf{regime 2}) and Theorem \ref{thequat1opt} (in case of \textbf{regime 3}). 
\end{proof}

The global behavior of the optimal instants is illustrated Figure \ref{FigExp2measures}. 
Both measures should be done as fast as possible for small $T$ (\textbf{regime 1}). When the duration $T$ is larger than a critical value $T^{(2)}_{2,\text{crit}}$, the instant of the second measure becomes distinct from zero and increases (\textbf{regime 2}). When the duration $T$ is larger than another critical value $T^{(2)}_{1,\text{crit}}$, the instant of the first measure becomes distinct from zero as well and increases (\textbf{regime 3}). Both optimal instants of measure exhibit continuity at the critical durations. This behavior is in accordance with Theorems \ref{ThCritereTnuls2mes}, \ref{T1sufficient} and \ref{thT1T2optMonotoneCont}.

\begin{figure}
\begin{center}
\input{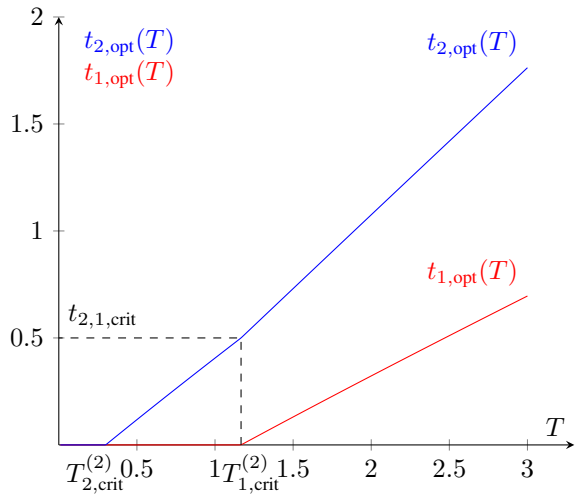}
\end{center}
\caption{$t^{(2)}_{1,\text{\rm opt}}(\Sig,T,v_0,v_1,v_2)$ and $t^{(2)}_{2,\text{\rm opt}}(\Sig,T,v_0,v_1,v_2)$ as functions of $T$ in a particular case. The parameters are $\Sig=1, v_0=v_1=v_2=1, T\in]0,5]$. In this example, the critical durations equal: $T^{(2)}_{2,\text{crit}}=0.3$ and $T^{(2)}_{1,\text{crit}}=\frac76$. For $T=\frac76$, the optimal schedule is $(0,0.5)$.}
\label{FigExp2measures}
\end{figure}

\subsection{The numerical algorithm of Coordinate Descent.}
\label{subsecCoordDescentDescription}

Theorem \ref{thequat1opt} provides a convenient theoretical description of the optimal schedule $(t^{(2)}_{1,\text{\rm opt}},t^{(2)}_{2,\text{\rm opt}})$ in \textbf{regime 3}. Let us look for an efficient algorithm for finding numeric values of these instants. The coordinate descent is proposed as such algorithm in this article. The first step of the coordinate descent is important as well in defining the $3$ regimes. 
This algorithm is described in Appendix \ref{Appendix}.

Updating $t_1$ is finding the minimum of a cost function $J(t_1)$ of a special type defined on a real interval. The golden-section search is used in this step. Some examples of functions this class are given Figure \ref{figExamplesT1toJ2meas}. It can be conjectured that all functions of this class are quasi-convex. If the function $J(t_1)$ is quasi-convex, the golden-section search is guaranteed to converge to the minimum of this function.

The cost function $J(t_1,t_2)$ is guaranteed to have only one CWLM, therefore the coordinate descent cannot converge to a point different from the global minimum of the function.

\subsection{The Experimental Performance of the Coordinate Descent.}
\label{subsecCoordDescent}
$100$ random runs of the algorithm have been performed in order to explore its convergence and the speed of convergence.

The parameters $\Sig=1$ and $T=10$ were fixed and the triples $(v_0,v_1,v_2)$ were chosen randomly from the region of the cube $[1,10]^3$, which corresponds to \textbf{regime 3}, according to the uniform distribution. More precisely, candidate points were chosen in $[1,10]^3$, then they were use in the experiment if they satisfied the condition of \textbf{regime 3}: $T^{(2)}_{1,\text{crit}}(\Sig,v_0,v_1,v_2)>T$.

The results are shown Figure \ref{figPerfsCordDescent}. They suggest an exponential convergence. Furthermore, the steps became shorter than $2\cdot10^{-6}$ after less than $10$ steps in all runs.

\begin{figure*}[p]
\begin{center}
\begin{tabular}{ccc}
\includegraphics[width=5cm]{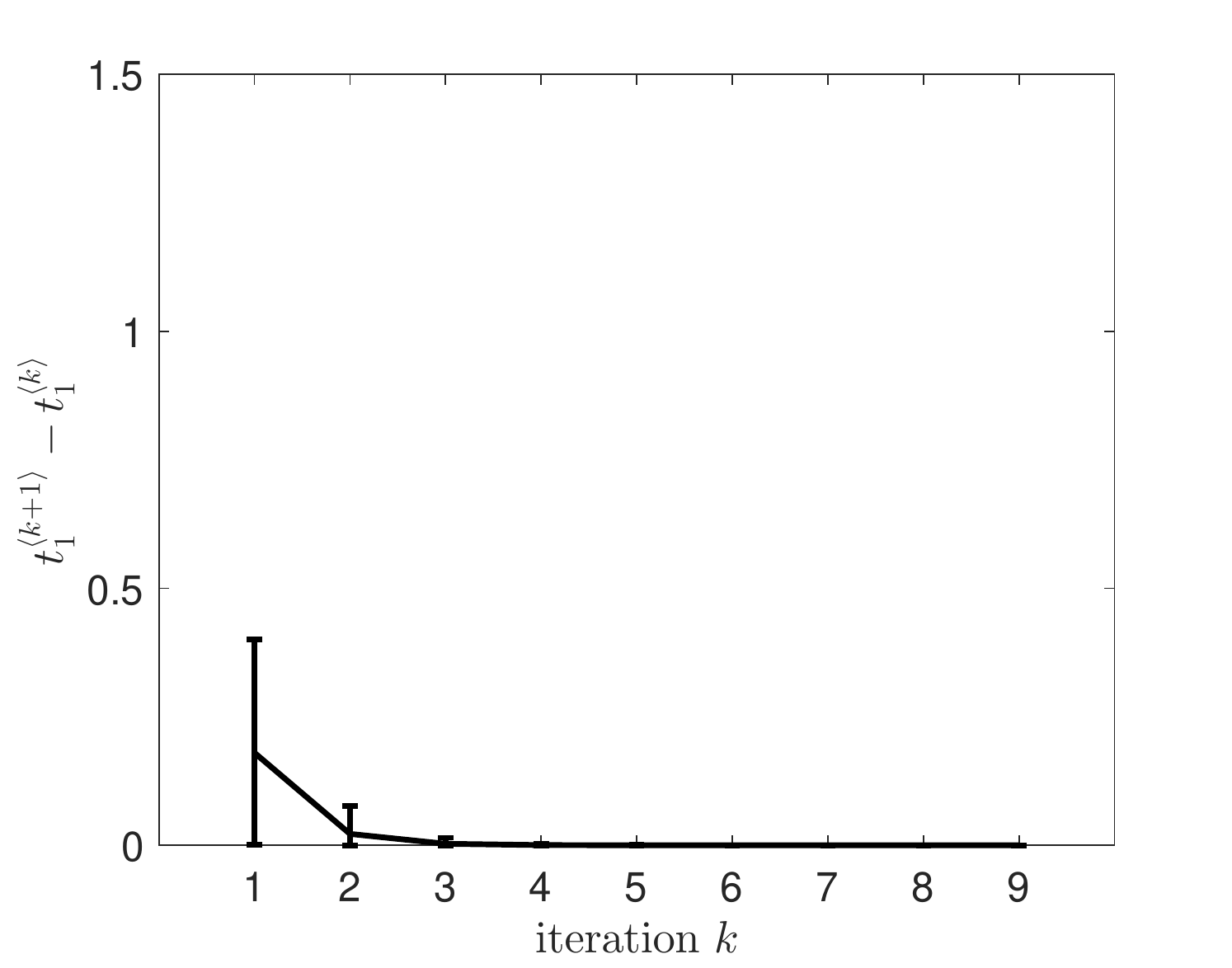}&\includegraphics[width=5cm]{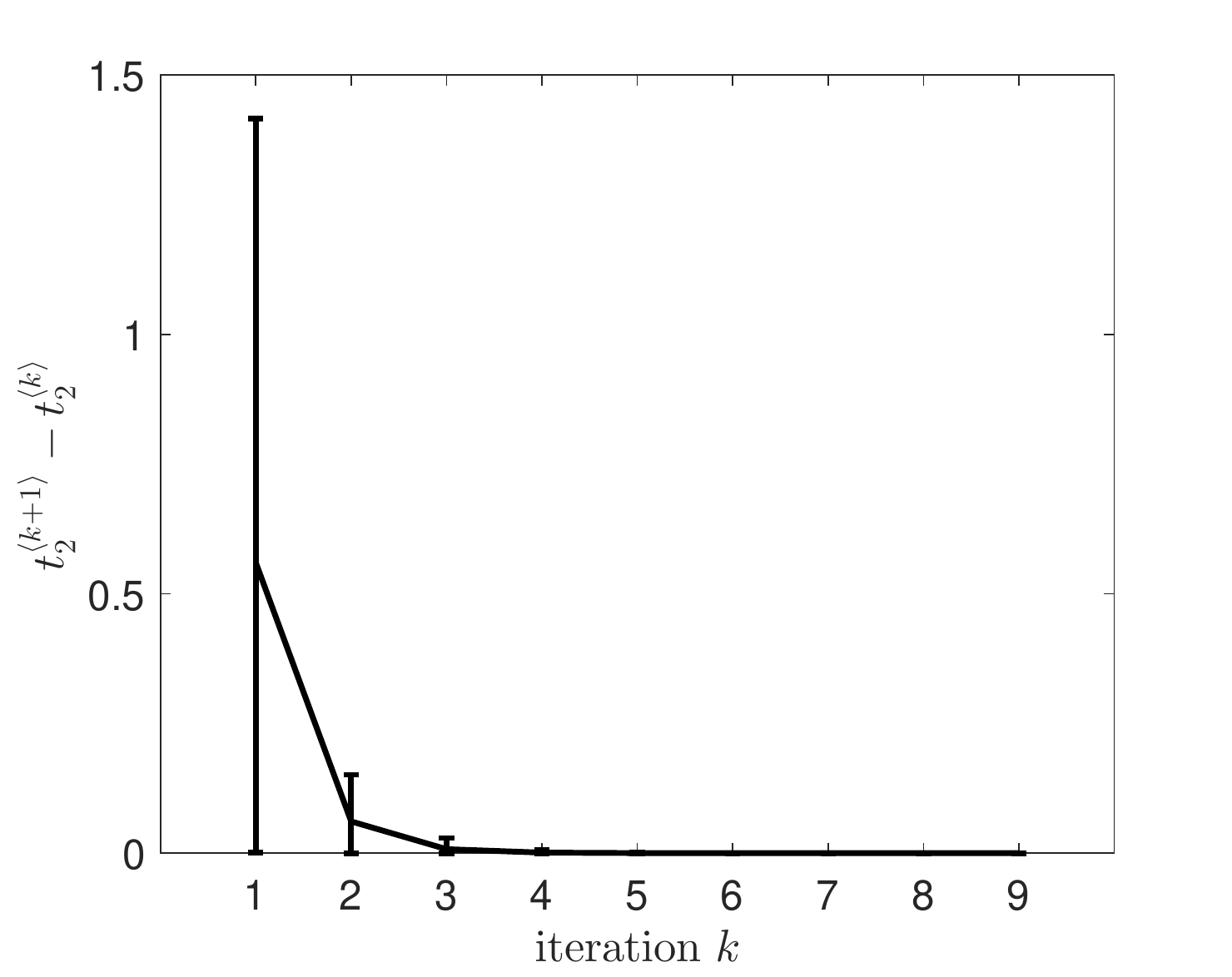}&\includegraphics[width=5cm]{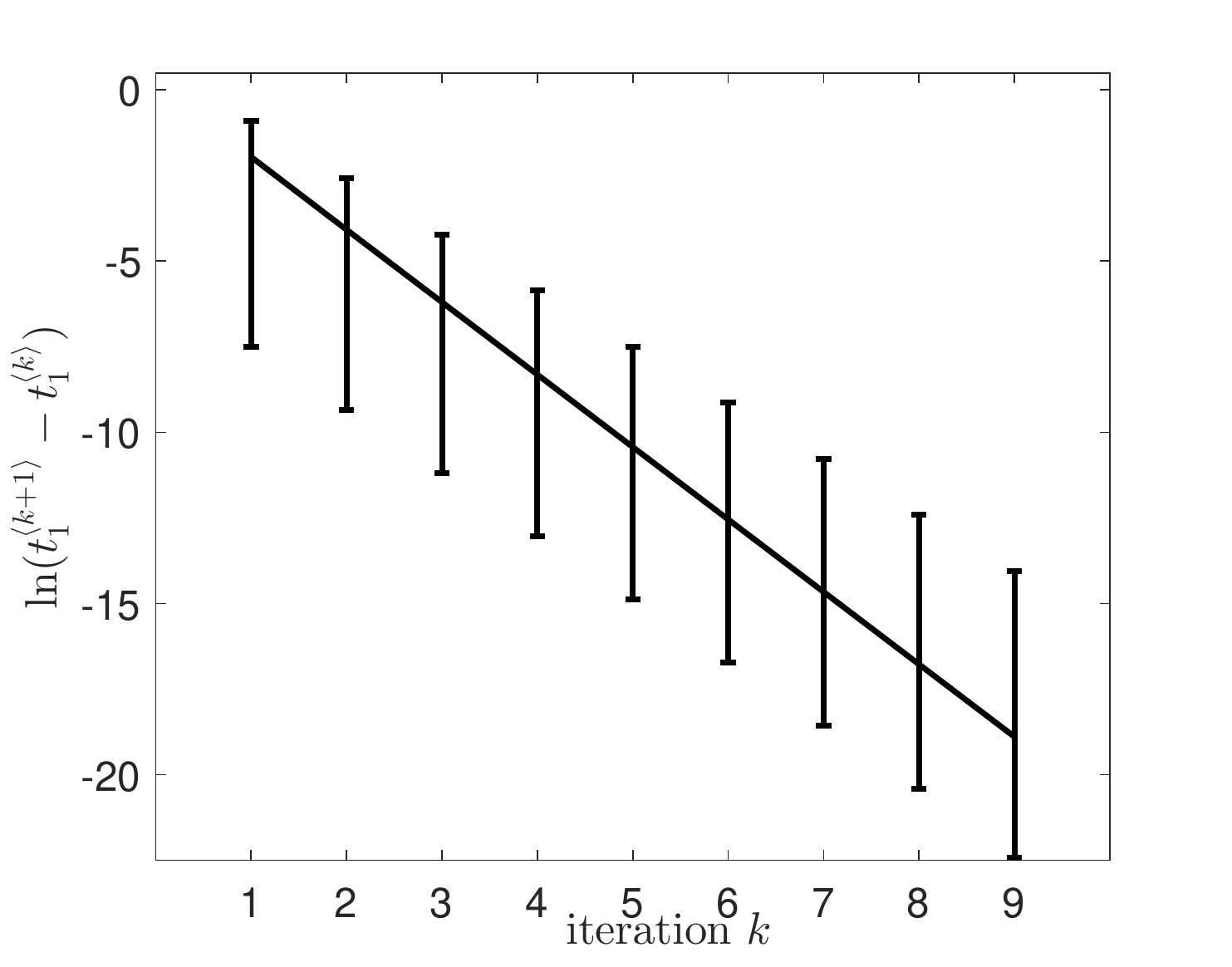}\\
\textbf{(a)}&\textbf{(b)}&\textbf{(c)}\\
\includegraphics[width=5cm]{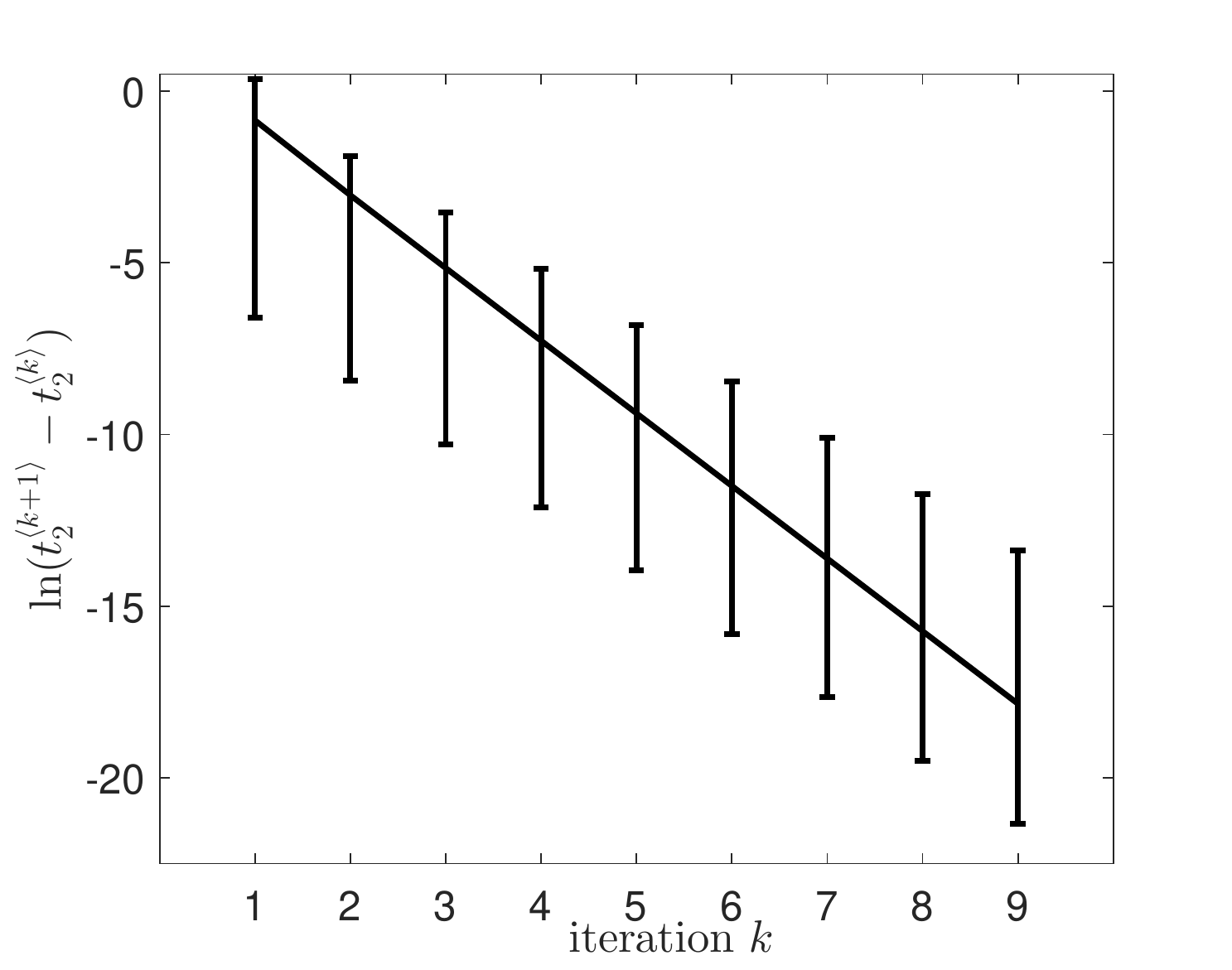}&\includegraphics[width=5cm]{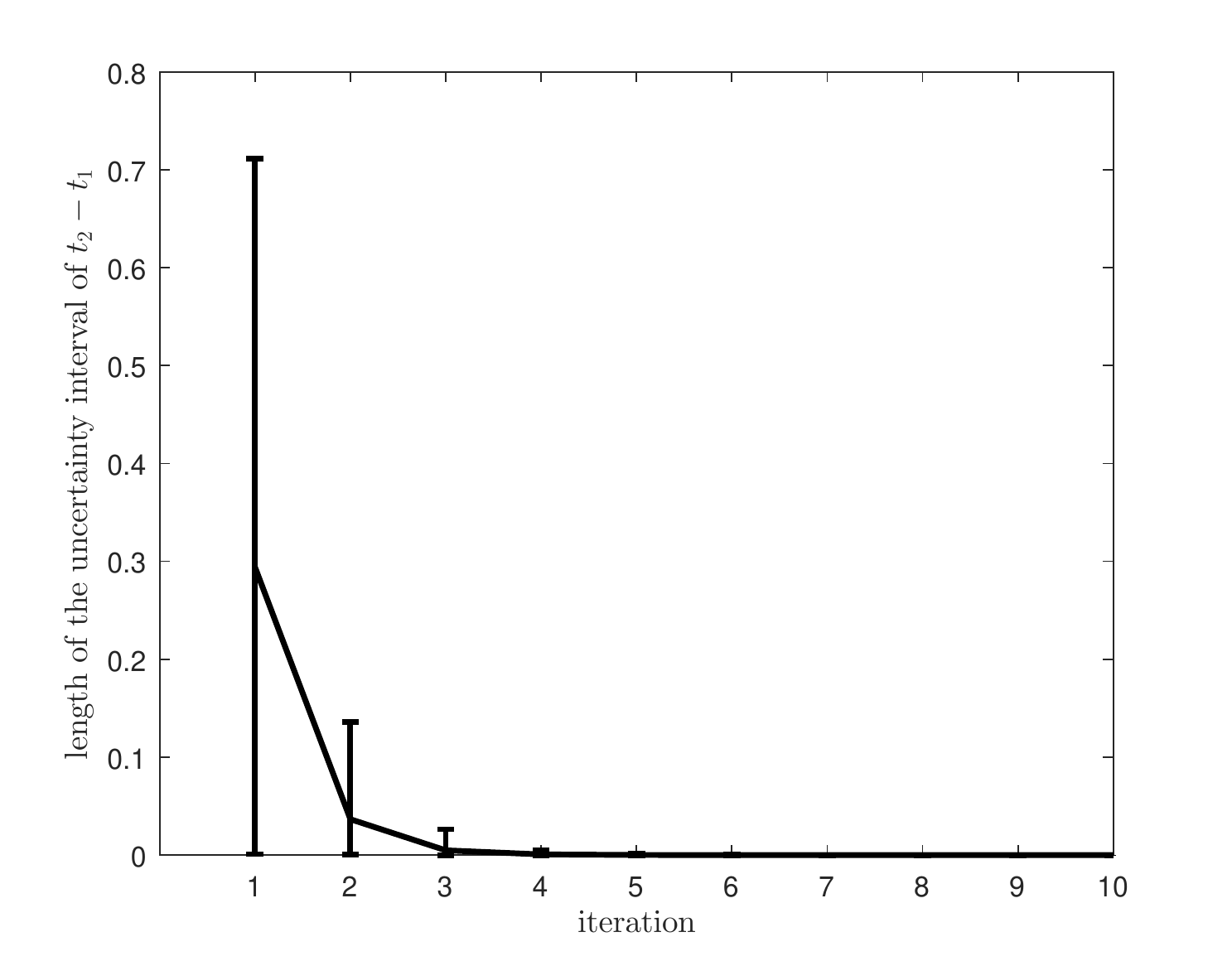}&\includegraphics[width=5cm]{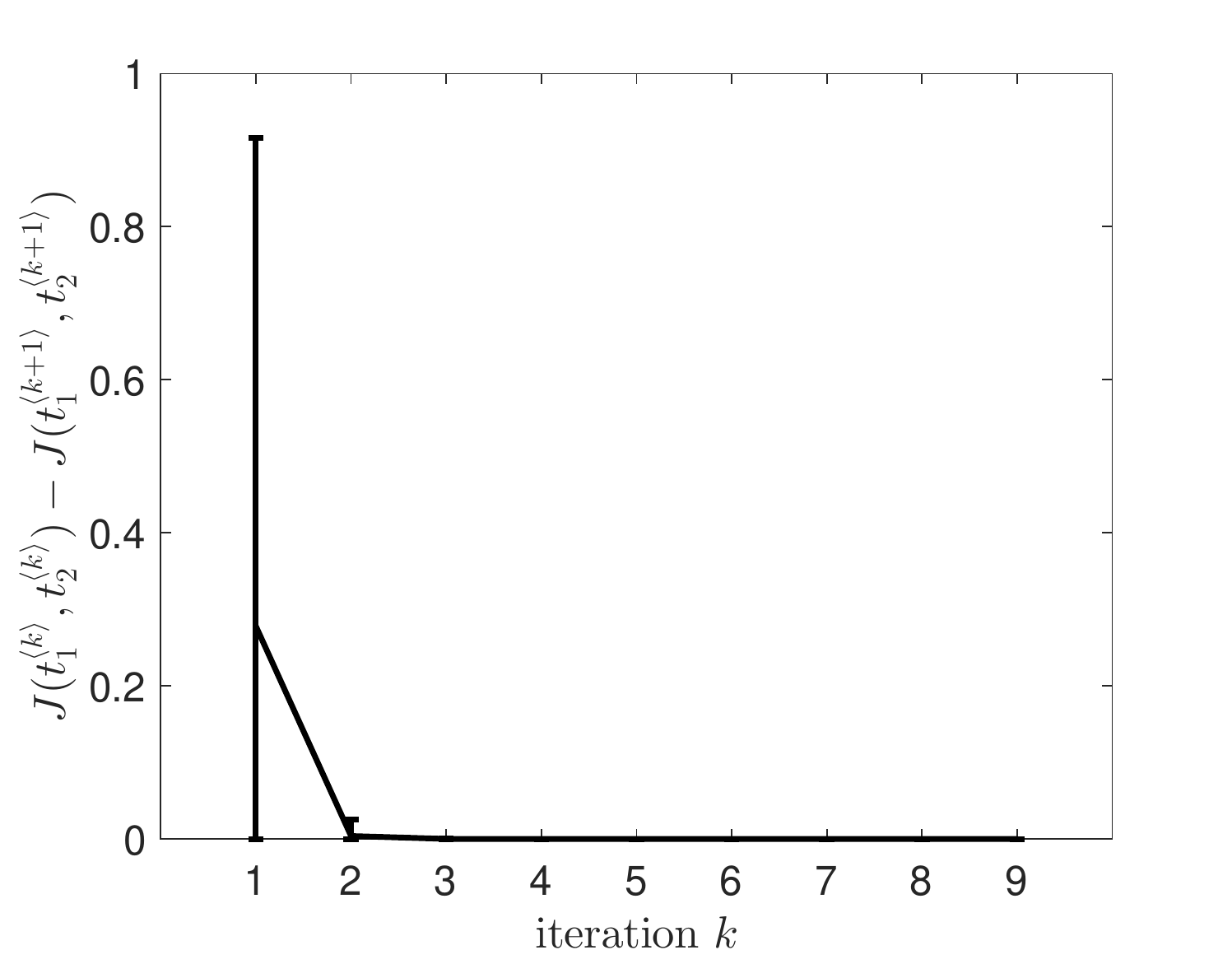}\\
\textbf{(d)}&\textbf{(e)}&\textbf{(f)}
\end{tabular}
\end{center}
\caption{Test performances of coordinate descent. The fixed parameters are $\Sig=1,T=10.$ $v_0,v_1,v_2$ have been drawn uniformly w.r.t. the Lebesgue measure from the part of the cube $[1,10]^3$ which corresponds to \textbf{regime 3}. \textbf{(a)} The increments of $t_1$. \textbf{(b)} The increments of $t_2$. \textbf{(c)} The natural logarithms of the increments of $t_1$. \textbf{(d)} The natural logarithms of the increments of $t_2$. \textbf{(e)} The difference $I^{(1)}_{\Sig,T,v_0,v_1,v_2}(t_1)-I^{(2)}_{\Sig,v_0,v_1,v_2}(t_1)$. According to Theorem \ref{thequat1opt}, the values of the functions $I^{(1)},I^{(2)}$ are estimations of the difference $t^{(2)}_{2,\text{\rm opt}}-t^{(2)}_{1,\text{\rm opt}}$ and they are equal only for the optimal value of $t_1$. \textbf{(f)} The decrements of the cost function $J$. The abscissa of every graph is the number of the step. The lines join the mean values of the corresponding quantities over all trials. The vertical error bars show the maxima and the minima.}
\label{figPerfsCordDescent}
\end{figure*}

\subsection{Comparison between the optimal and the regular schedules.}
\label{subsecComparer}

A numerical experiment of estimation of the gain of the optimal schedule compared to the intuitive sampling $(\frac{T}3,\frac{2T}3)$ has been done. The optimal schedules $(t^{(2)}_{1,\text{opt}},t^{(2)}_{2,\text{opt}})$ have been computed together with the associated costs $J_\text{opt}$ for $\Sig=1,T=1,v_0\in\{0,2,5\}$ and $v_1,v_2$ varying from $0$ to $5$. The costs $J_\text{reg}$ achieved with the regular sampling have been computed as well. Figure \ref{figGains2Measures} shows three contour plots of the gain $\frac{J_\text{reg}-J_\text{opt}}{J_\text{reg}}$ as functions of $v_1,v_2$.

\begin{figure*}[p] %[p] : 16/1
\begin{center}
\begin{tabular}{ccc}
\includegraphics[width=5cm]{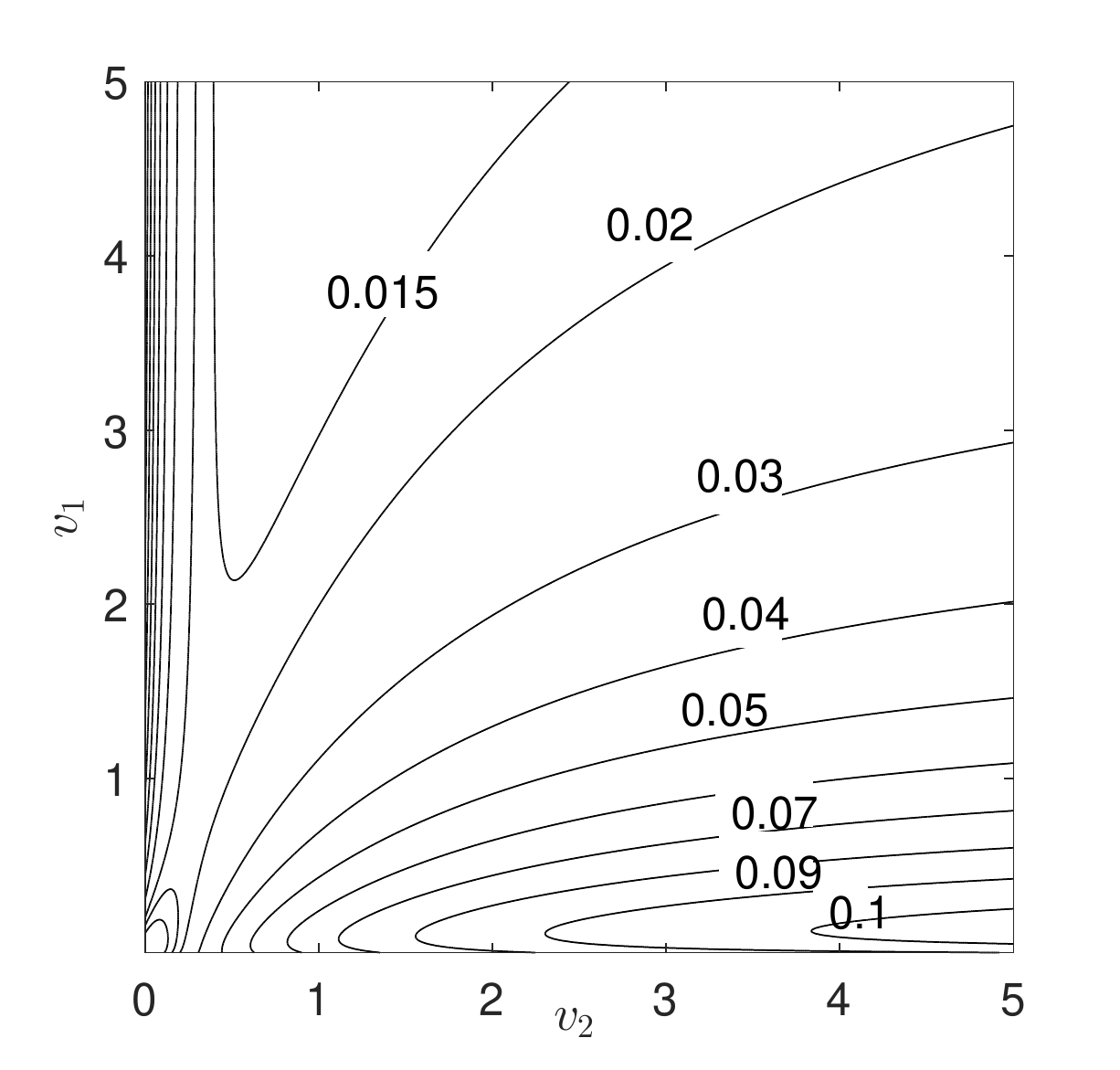}&\includegraphics[width=5cm]{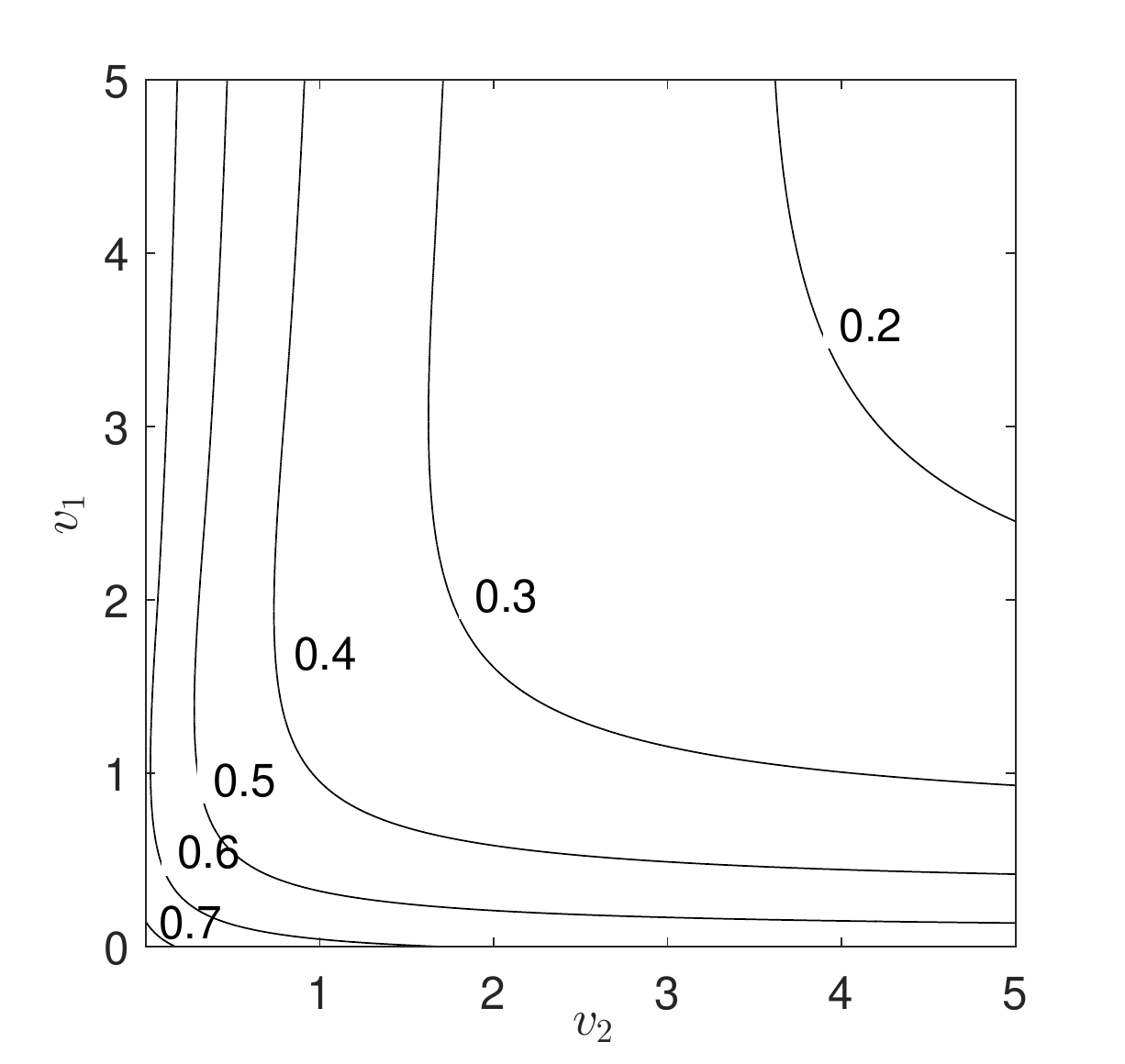}&\includegraphics[width=5cm]{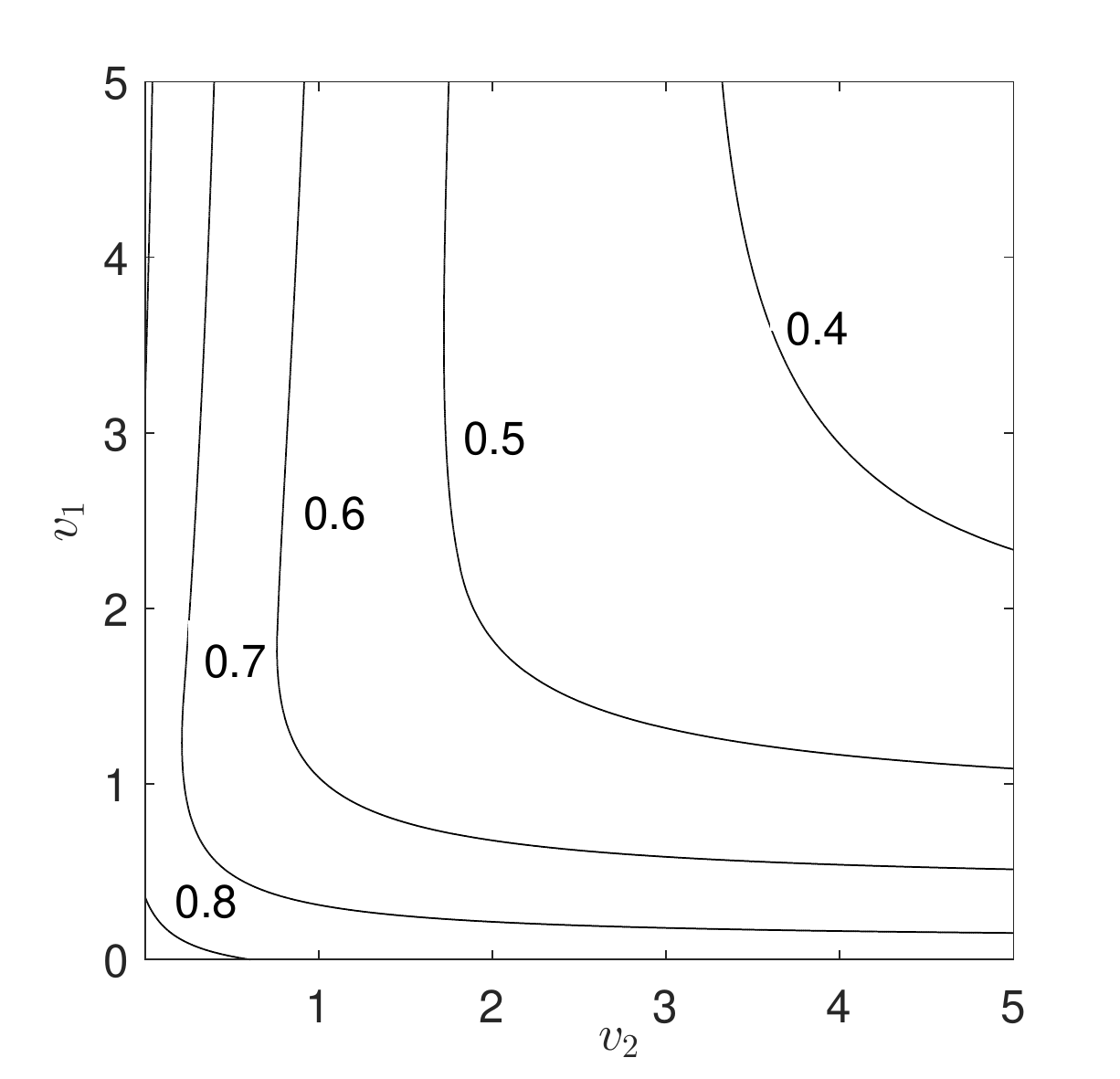}\\
\textbf{(a)}&\textbf{(b)}&\textbf{(c)}
\end{tabular}
\end{center}
\caption{The relative gain achieved by the optimal schedule as functions of $v_1,v_2\in[0,5]$. \textbf{(a)} $v_0=0$. The contours correspond to the values $0.01,0.015$, and $0.02$ to $0.1$ by steps of $0.01$. \textbf{(b)} $v_0=2$. \textbf{(c)} $v_0=5$.}
\label{figGains2Measures}
\end{figure*} 

These figures can be compared with the gain in case of $1$ measure (Figure \ref{figComparisonRegular}). For parameters in the considered range, the gain can reach up to $86\%$.

\section{Conclusion and Perspectives} 
\label{secPerspectives}
Sampling strategies for a phenomenon of finite length have been investigated under the assumption that the phenomenon can only be measured a small number of times by instruments with different properties (error variances). Irregular sampling can lead to a substantial gain in mean error variance of the estimator.

The assumption of a small number of available measures can be satisfied if the process itself is short or the measurement devices have a limited (and non-renewable) physical resource, e.~g., \cite{pinguins}. This can also happen if each measure is expensive.

A simple model is studied, where  the variance about the system parameters (here a single parameter) evolving over a finite period of time grows linearly in the absence of measure. The properties of the optimal measure timetable according to the criterion of minimization of the mean variance are considered. 

In Section \ref{sec1meas}, the particular case, where the instant of exactly $1$ measure is to be chosen, is studied in detail. Section \ref{sec2meas} is devoted to the particular case, where the instants of $2$ measures are to be chosen. 

The system can behave in different regimes. When the duration of the process is short, it is optimal to take all measures at the moment zero. If it is larger, than a critical value, one optimal instant of measure moves from zero to the inside of the interval. In the case of one measure, there is one critical duration, while in case of two measures there are two.

It is proved that the critical durations in case of 1 or 2 measures are increasing functions of $v_0$. This corresponds to a simple intuition: the larger $v_0$ is, the less exact is the information, the higher are chances that it should be supported by a measure. This corresponds to the intuition stated in the introduction: in the optimal sampling, the more precise measure may be made shortly after the less precise one.

The computations relative to the case $n=2$ (shown Figure \ref{figGains2Measures}) suggest that when $v_0\ll\Sig T,v_1,v_2$, the gain in comparison with the regular schedule is modest. On the other hand, it increases if the variance $v_0$ of the initial information increases or if the variances $v_1,v_2$ of the measures are very different. The first conclusion is also confirmed experimentally in case of $1$ measure (see Figure \ref{figComparisonRegular} \textbf{(b)}).

A setting, where a large number of measures are made under a constraint of periodic ``windows", is considered Section \ref{subsecIterated1meas}. The instants of measurement are determined using local optimization. It is shown that local optimization leads to regular sampling (Theorem \ref{thPeriodicFromAperiodic}) when the  number of measures is large. This result suggests that the global optimization is necessary for getting an improvement of performance.

One goal of the future research is to find the optimal (in the sense of the cost function \eqref{defCost}) measurement instants when the number of measures is $n>2$. The methods of this article can be adapted. Some qualitatively new conjectures also appear from the experiments in this setting. Allowing the number of measures to vary is another possible development of the results presented here.

In this problem, the order of the measures is fixed. It is also possible to allow it to vary. The main property of this problem is the fact that the cost function is no longer rational, but piecewise-rational. 

Another objective of the future research is to consider more complex models than the real Brownian motion considered presently.

\newpage   

\appendix[Pseudo-code of the coordinate descent algorithm.]
\label{Appendix}

\begin{figure}[H]
\begin{algorithmic}[1]
\IF {\(T\leqslant \frac{v_{0,1}}{\Sig\left(\frac{v_2}{v_{0,1}+v_2}+1\right)}\)}
\RETURN \((0,0)\) \textit{(regime 1)}
\ELSE
\STATE \(t_2^{\langle1\rangle}:=t_{\text{opt}}^{(1)}(\Sig,T,v_{0,1},v_2)\)
\IF {\(A(v_0,v_1,v_2)(\Sig t_2^{\langle1\rangle})^3+B(v_0,v_1,v_2)(\Sig t_2^{\langle1\rangle})^2+C(v_0,v_1,v_2)\Sig t_2^{\langle1\rangle}+D(v_0,v_1,v_2)\geqslant 0\)}
\RETURN \((0,t_2^{\langle1\rangle})\) \textit{(regime 2)}
\ELSE
\STATE Initialization \textit{(regime 3, coordinate descent)}
\STATE \(t_2:=t_2^{\langle1\rangle}\)%t1opt(....)
\STATE \(t_1:=t_1^{\langle1\rangle}=\argmin_{t_1}J_{\Sig,T,v_0,v_1,v_2}(t_1,t_2^{\langle1\rangle})\)
\REPEAT %(the number of the loop $k$ starts from $2$)
\STATE \(t_2:=t_2^{\langle k\rangle}=t_1^{\langle k-1\rangle}+t_{\text{opt}}^{(1)}(\Sig,T-t_1^{\langle k-1\rangle},(v_0+\Sig t_1^{\langle k-1\rangle})\parallelsum v_1,v_2)\) \STATE \(t_1:=t_1^{\langle k\rangle}=\argmin_{t_1}J_{\Sig,T,v_0,v_1,v_2}(t_1,t_2^{\langle k\rangle})\)
\UNTIL convergence
\ENDIF
\ENDIF
\end{algorithmic}
\caption{Compute the optimal instants of $2$ measures. Arguments: $\Sig,T,v_0,v_1,v_2\in\R_+^*$.}
\label{AlgoCoordDescent}
\end{figure}

\end{document}